\newcommand\relativepath{}

\documentclass[11pt, a4paper]{article}
\usepackage{fullpage}
\usepackage{algorithm}

\newcommand\bigone[1]{}
\newcommand\smallone[1]{#1}
\usepackage{amsfonts,amssymb,amsmath,amsthm,latexsym}
\usepackage{array,xspace}

\newcommand{\ignore}[1]{}

\newcommand{\eps}{\varepsilon}

\newcommand{\etal}{{\em et al.}\xspace}

\def\makeletter#1{%
\expandafter \newcommand \csname b#1\endcsname {\mathbb{#1}}%
\expandafter \newcommand \csname c#1\endcsname {\mathcal{#1}}%
\expandafter \newcommand \csname t#1\endcsname {\widetilde{#1}}%
\expandafter \newcommand \csname ct#1\endcsname {\widetilde{\mathcal{#1}}}%
}
\def\makeletters(#1#2){\makeletter#1\ifx#2.\else\makeletters(#2)\fi}
\makeletters(QWERTYUIOPASDFGHJKLZXCVBNM.)

\def\makeSkob#1#2#3{%
\def\LLL{\left} \def\RRR{\right}
\expandafter \edef \csname #1\endcsname #2##1#3{\SkobInner}
\def\LLL{\bigl} \def\RRR{\bigr}
\expandafter \edef \csname #1A\endcsname #2##1#3{\SkobInner}
\def\LLL{\Bigl} \def\RRR{\Bigr}
\expandafter \edef \csname #1B\endcsname #2##1#3{\SkobInner}
\def\LLL{\biggl} \def\RRR{\biggr}
\expandafter \edef \csname #1C\endcsname #2##1#3{\SkobInner}
\def\LLL{\Biggl} \def\RRR{\Biggr}
\expandafter \edef \csname #1D\endcsname #2##1#3{\SkobInner}
\def\LLL{} \def\RRR{}
\expandafter \edef \csname #1O\endcsname #2##1#3{\SkobInner}
}

\def\SkobInner{\LLL(##1\RRR)} \makeSkob{s}[]
\def\SkobInner{\LLL[##1\RRR]} \makeSkob{sk}[]
\def\SkobInner{\LLL\lbrace##1\RRR\rbrace} \makeSkob{sfig}{}{}
\def\SkobInner{\LLL\lfloor##1\RRR\rfloor} \makeSkob{floor}[]
\def\SkobInner{\LLL\lceil##1\RRR\rceil} \makeSkob{ceil}[]
\def\SkobInner{\LLL\langle##1\RRR\rangle} \makeSkob{ip}<>
\def\SkobInner{\LLL|##1\RRR\rangle} \makeSkob{ket}|>
\def\SkobInner{\LLL|##1\RRR|} \makeSkob{abs}||
\def\SkobInner{\LLL\|##1\RRR\|} \makeSkob{norm}||
\def\SkobInner{\LLL\|##1\RRR\|_{\noexpand\mathrm F}} \makeSkob{normFrob}||
\def\SkobInner{\LLL\|##1\RRR\|_{\noexpand\mathrm{tr}}} \makeSkob{normtr}||

\newcommand{\midA}{\mathbin{\bigl|}}

\def \elem[#1]{[\![#1]\!]}
\def \bigfrac#1/{\left.#1\right/}
\def \bigfracR/#1.{\left/#1\right.}

\newcommand{\pfstart}{\begin{proof}} 
\newcommand{\pfsketch}{\begin{proof}[Proof sketch]}
\newcommand{\pfend}{\end{proof}} 
\newcommand{\itemstart}{\begin{itemize}\itemsep0pt}
\newcommand{\itemend}{\end{itemize}}
\newcommand{\descrstart}{\begin{description}\itemsep0pt}
\newcommand{\descrend}{\end{description}}
\newcommand{\enumstart}{\begin{enumerate}\itemsep0pt}
\newcommand{\enumend}{\end{enumerate}}

\newcommand{\negmedskip}{\vspace{-\medskipamount}}
\newcommand{\negbigskip}{\vspace{-\bigskipamount}}

\newcommand{\CrossRef}[2]{#1~\ref{#2}}

\bigone{\newtheorem{thm}{Theorem}[chapter]}
\smallone{\newtheorem{thm}{Theorem}}

\newcommand{\maketheorem}[2]{
\newtheorem{#1}[thm]{#2}
\expandafter\def \csname ref#1\endcsname ##1{\CrossRef{#2}{#1:##1}}
}

\maketheorem{lem}{Lemma}
\maketheorem{prp}{Proposition}
\maketheorem{cor}{Corollary}
\maketheorem{clm}{Claim}
\maketheorem{fact}{Fact}

\theoremstyle{definition}
\maketheorem{rem}{Remark}
\maketheorem{obs}{Observation}
\maketheorem{defn}{Definition}
\maketheorem{exm}{Example}
\maketheorem{assump}{Assumption}

\def \rf(#1:#2){\csname ref#1\endcsname{#2}}
\def \rfitem(#1@#2){\rf(#2)(\ref{#1@#2})}
\def \rfitemE(#1@#2){\ref{#2}(\ref{#1@#2})}

\def\mycommand#1#2{
\expandafter\newcommand \csname#1\endcsname {#2}%
}
\def\mytxtcommand#1#2{
\expandafter\newcommand \csname#1\endcsname {#2\xspace}%
}
\def\remycommand#1#2{
\expandafter\renewcommand \csname#1\endcsname {#2}%
}

\newcommand\draft[1]{}
\newcommand\release[1]{#1}

\usepackage{tikz}
\usepackage{tikz-qtree}

\usetikzlibrary{arrows}

\release{\usepackage[breaklinks]{hyperref}}

\title{Separations in Query Complexity Based on Pointer Functions}
\author{
Andris Ambainis$^1$ \and
Kaspars Balodis$^1$ \and
Aleksandrs Belovs$^2$ \and
Troy Lee$^3$ \and
Miklos Santha$^4$ \and
Juris Smotrovs$^1$}

\date{}

\begin{document}
\maketitle
\footnotetext[1]{Faculty of Computing, University of Latvia, ({\tt andris.ambainis@lu.lv, kbalodis@gmail.com,  juris.smotrovs@lu.lv}).}
\footnotetext[2]{CWI, the Netherlands, ({\tt stiboh@gmail.com}).} 
\footnotetext[3]{School of Physical and Mathematical Sciences, Nanyang Technological University and Centre for Quantum Technologies and MajuLab, UMI 3654, Singapore ({\tt troyjlee@gmail.com}).} 
\footnotetext[4]{LIAFA, Univ. Paris 7, CNRS, 75205 Paris, France;  and
Centre for Quantum Technologies, National University of Singapore,
Singapore 117543 ({\tt miklos.santha@liafa.univ-paris-diderot.fr}).}

\mytxtcommand{goos}{G\"o\"os}
\mytxtcommand{goosfunction}{\goos-Pitassi-Watson function}
\newcommand{\Goosfunction}{$g_{n,m}\ $}

\mycommand{NAND}{\mathrm{NAND}}
\mycommand{tdeg}{\mathop{\widetilde{\mathrm{deg}}}}
\mycommand{bool}{\{0,1\}}
\mycommand{cube}{\bool^n}

\begin{abstract}
In 1986, Saks and Wigderson conjectured that the largest separation between deterministic and zero-error randomized 
query complexity for a total boolean function is given by the function $f$ on $n=2^k$ bits defined by a complete binary tree 
of NAND gates of depth $k$, which achieves $R_0(f) = O(D(f)^{0.7537\ldots})$.  We show this is false by giving an example 
of a total boolean function $f$ on $n$ bits whose deterministic query complexity is $\Omega(n/\log(n))$ while its zero-error 
randomized query complexity is $\tO(\sqrt{n})$.  
We further show that the quantum query complexity of the same function is $\tO(n^{1/4})$, giving the first example of a total function with a super-quadratic gap between its quantum and deterministic query complexities.  

We also construct a total boolean function $g$ on $n$ variables that has zero-error randomized query complexity 
$\Omega(n/\log(n))$ and bounded-error randomized query complexity $R(g) = \tO(\sqrt{n})$.  This is the first super-linear separation between these two complexity measures.  
The exact quantum query complexity of the same function is $Q_E(g) = \tO(\sqrt{n})$.  

These two functions show that the relations $D(f) = O(R_1(f)^2)$ and $R_0(f) = \tO(R(f)^2)$ are optimal, up to poly-logarithmic factors.  
Further variations of these functions give additional separations between other query complexity measures: a cubic separation between $Q$ and $R_0$, a $3/2$-power separation between $Q_E$ and $R$, and a 4th power separation between approximate degree and bounded-error randomized query complexity.

All of these examples are variants of a function recently introduced by \goos, Pitassi, and Watson 
which they used to separate the unambiguous 1-certificate complexity 
from deterministic query complexity and to resolve the famous Clique versus Independent Set problem in communication complexity. 
\end{abstract}

\newpage
\section{Introduction}
Query complexity has been very useful for understanding the power of different computational models. 
In the standard version of the query model, we want to compute a boolean function $f\colon \bool^n \rightarrow \bool$ on an initially unknown input 
$x \in \bool^n$ that can only be accessed by asking queries of the form $x_i = ?$.  
The advantage of query complexity is that we can often prove tight lower bounds and have provable separations between different 
computational models.  This is in contrast to the Turing machine world where lower bounds and separations between complexity classes often have 
to rely on unproven assumptions. At the same time, the model of query complexity is simple and captures the essence of quite a few
natural computational processes. 

We use $D(f), R(f)$ and $Q(f)$ to denote the minimum number of queries in deterministic, randomized and quantum query algorithms\footnote{By default,
we use $R(f)$ and $Q(f)$ to refer to bounded-error algorithms (i.e., algorithms that compute $f(x)$ correctly on every input $x$ with probability at
least $9/10$).} that compute $f$.
It is easy to see that $Q(f) \le R(f) \le D(f)$ for any function $f$.  For partial functions (that is, functions whose domain is a 
strict subset of $\bool^n$), huge separations are known between all these measures.
For example, a randomized algorithm can tell if an $n$-bit boolean string has $0$ ones or 
at least $n/2$ ones with a constant number of queries, while any deterministic algorithm requires $\Omega(n)$ queries 
to do this.  Similarly, Aaronson and Ambainis~\cite{aaronson:forrelation} recently constructed a partial boolean function 
$f$ on $n$ variables that can be evaluated using one quantum query but requires
$\Omega(\sqrt{n})$ queries for randomized algorithms.

The situation is quite different for total functions.\footnote{In the rest of the paper we will exclusively talk about total functions. Hence, we sometimes drop this qualification.}
Here it is known that $D(f), R(f)$, and $Q(f)$ are all polynomially related.  In fact, $D(f) = O(R(f)^3)$ \cite{nisan:bs} and 
$D(f) = O(Q(f)^6)$ \cite{beals:pol}.  
A popular variant of randomized algorithms is the zero-error (Las Vegas) model in which a randomized algorithm always has to
output the correct answer, but the number of queries after which it stops can depend on the algorithm's coin flips.
The complexity $R_0(f)$ is defined as the expected number of queries, over the randomness of the algorithm, for the 
worst case input $x$. 
A tighter relation $D(f) \le R_0(f)^2$ is known for Las Vegas algorithms 
(this was independently observed by several authors \cite{HH87,BI87,T90}).  
Nisan has even shown $D(f) = O(R_1(f)^2)$ \cite{nisan:bs},
where $R_1(f)$ is the one-sided error randomized complexity of $f$.
Recently, Kulkarni and Tal~\cite{kulkarni:fractionalBS}, basing on a result by Midrij\=anis~\cite{midrijanis:rand}, showed that $R_0(f) = \tO(R(f)^2)$, where the $\tO$ notation hides 
poly-logarithmic factors.

While it has been widely conjectured that these relations are not tight, little progress has been made in the past 20 
years on improving these upper bounds or exhibiting functions with separations approaching them.  
Between $D(f)$ and $R_0(f)$, the best separation known for a 
total function is the function $\NAND^k$ on $n=2^k$ variables defined by a complete binary NAND tree of depth $k$.  This 
function satisfies $R_0(\NAND^k) = O(D(\NAND^k)^{0.7537...})$ \cite{snir:nand}.  Saks and Wigderson 
showed that this upper bound is optimal for $\NAND^k$, and conjectured that this is the largest gap 
possible between $R_0(f)$ and $D(f)$ \cite{saks:nand}.  This function also provides the largest 
known gap between $R(f)$ and $D(f)$, and satisfies
$R(\NAND^k)= \Omega(R_0(\NAND^k))$ \cite{San95}.  This situation points to the broader fact that, as far as we are 
aware, no super-linear gap 
is known between $R(f)$ and $R_0(f)$ for a total function $f$.  Between $Q(f)$ and $D(f)$, the 
largest known separation is quadratic, given by the OR function on $n$ bits, which satisfies $Q(f) = O(\sqrt{n})$ \cite{grover:search} and $D(f) = \Omega(n)$.
 
\subsection{Our results}
We improve the best known separations between all of these measures.  In particular, we show that
\begin{itemize}
  \item There is a function $f$ with $R_0(f) = \tO(D(f)^{1/2})$.  This refutes the nearly 30 year old conjecture of Saks and 
  Wigderson \cite{saks:nand}, and shows that the upper bounds $D(f) \le R_0(f)^2$ and $D(f)  = O(R_1(f)^2)$ are tight, up to poly-logarithmic factors.  
  \item There is a function $f$ with $R_1(f) = \tO(R_0(f)^{1/2})$.  
	This is also nearly optimal due to Nisan's result $D(f) = O(R_1(f)^2)$, as well as the upper bound 
	$R_0(f) = \tO(R(f)^2)$ by Kulkarni and Tal.
	Previously, no super-linear separation was known even between $R(f)$ and $R_0(f)$.
  \item There is a function $f$ with $Q(f) = \tO(D(f)^{1/4})$.  This is the first improvement in nearly 20 years to the quadratic separation given by Grover's search algorithm \cite{grover:search}.
  \item Let $Q_E(f)$ be the exact quantum query complexity, the minimal number of queries needed by a quantum algorithm that stops after a fixed number of steps and outputs $f(x)$ with probability 1.  We exhibit functions $f_1,f_2$ for which $Q_E(f_1) = \tO(R_0(f_1)^{1/2})$ 
  and $Q_E(f_2) = \tO(R(f_2)^{2/3})$.  This improves the best known separation of Ambainis from 2011 \cite{ambainis:exact} giving an $f$ for which $Q(f) = O(R(f)^{0.867\ldots})$.
	Prior to the work of Ambainis, no super-linear separation was known, the largest known separation being a factor of 2, attained for the PARITY function~\cite{cleve:phaseEstimation}.
	
\end{itemize}
A full list of our results are given in the following table.  Subsequent to our work, Ben-David \cite{ben-david:super-grover} has additionally given a super-quadratic separation between $Q(f)$ and $R(f)$, exhibiting a function with $Q(f) = \tO(R(f)^{2/5})$. 
\[
\label{eq:table}
\begin{array}{l|l@{}l|l@{}l|lll|}
&\multicolumn{2}{c|}{\text{lower bound for all $f$}} &\multicolumn{2}{c|}{\text{previous separation}} &\text{this paper} &  \text{function}& \multicolumn{1}{c|}{\text{result}} \\
\hline
R_0(f)& \Omega(D(f)^{1/2})& \cite{HH87,BI87,T90} & O(D(f)^{0.753\ldots})& \cite{snir:nand} & \tO(D(f)^{1/2})& 
f_{2n,n}& \text{\rf(cor:R0-D)} \\ 
Q(f) & \Omega(D(f)^{1/6})& \cite{beals:pol} & O(D(f)^{1/2})&\cite{grover:search} & \tO(D(f)^{1/4}) & 
f_{2n,n} & \text{\rf(cor:R0-D)} \\
R_1(f) & \Omega(R_0(f)^{1/2})& \cite{nisan:bs} & O(R_0(f)) && \tO(R_0(f)^{1/2}) & 
g_{n,n} &\text{\rf(cor:R1-R0)} \\
Q_E(f) & \Omega(R_0(f)^{1/3})& \cite{midrijanis:exact} & O(R_0(f)^{0.867\ldots}) &\cite{ambainis:exact}& \tO(R_0(f)^{1/2}) & g_{n,n} &\text{\rf(cor:R1-R0)} \\
Q(f) & \Omega(R_0(f)^{1/6})& \cite{beals:pol} & O(R_0(f)^{1/2})&\cite{grover:search} & \tO(R_0(f)^{1/3}) & h_{n,n,n^2} &\text{\rf(cor:Q-R0)} \\
Q_E(f) & \Omega(R(f)^{1/3})& \cite{midrijanis:exact} & O(R(f)^{0.867\ldots}) &\cite{ambainis:exact}& \tO(R(f)^{2/3}) & 
h_{1,n,n^2} & \text{\rf(cor:QE-R2)} \\
\tdeg(f) & \Omega(D(f)^{1/6})& \cite{beals:pol} & O(R(f)^{1/2})&\cite{nisan:pol} & \tO(R(f)^{1/4}) & h_{1,n,n^2} &\text{\rf(cor:deg-R2)}\\
\hline
\end{array}
\]

Other separations can be obtained from this table using relations between complexities in \rf(fig:complexities).

\subsection{\goosfunction}
All of our separations are based on an amazing function recently introduced by G\"o\"os, Pitassi, and Watson~
\cite{goos:partitionNumber} to
resolve the deterministic communication complexity of the Clique vs. Independent set problem, thus solving a long-standing open problem in communication complexity.  They solved this problem by
first solving a coresponding question in the query complexity model and
then showing a general ``lifting theorem'' that lifts the 
hardness of a function in the deterministic query model to the hardness of a derived function in the model of deterministic 
communication complexity. In the query complexity model, their goal was to exhibit a total boolean 
function $f$ that has large deterministic query complexity and small unambiguous 1-certificate complexity.\footnote{A 
subcube is the set of strings consistent with 
a partial assignment $x_{i_1}=b_1, \ldots, x_{i_s}=b_s$.  Its length is $s$, the number of assigned variables.  The 
unambiguous 1-certificate complexity is the smallest $s$ such that $f^{-1}(1)$ can be partitioned into subcubes of 
length $s$ (whose corresponding partial assignments are consequently 1-certificates of $f$).  
See \rf(sec:prelim) for full definitions.}

The starting point of their construction is the boolean function $f\colon\bool^M\to\bool$ with the input variables $x_{i,j}$ arranged in a rectangular grid $M=[n]\times [m]$.
The value of the function is 1 if and only if there exists a unique all-1 column.
The deterministic complexity of this function is $\Omega(nm)$, since it is hard to distinguish an input with precisely one zero in each column from the input in which one of the zeroes is flipped to one.
It is also easy to construct a 1-certificate of length $n+m-1$: Take the all-1 column and one zero from each of the remaining columns.  This certificate is not always unique, however, as there can be multiple zeroes in a column and any of them can be chosen in a certificate.  Indeed, it is impossible to \emph{partition} the set of all positive inputs into subcubes of small length.

%The goal behind the \goosfunction was to exhibit a total boolean function $f$ that has large deterministic query complexity and small unambiguous 1-certificate complexity (see \rf(sec:prelim) for the definition).
%Consider first the boolean function $f\colon\bool^M\to\bool$ with the input variables $x_{i,j}$ arranged in a rectangular grid $M=[n]\times [m]$.
%The value of the function is 1 if and only if \todo{expanded iff} there exists a unique all-1 column.
%The complexity of this function is $\Omega(nm)$, since it is hard to distinguish an input with precisely one zero in each column from the input in which one of the zeroes is flipped to one.
%It is also easy to construct a 1-certificate of length $n+m-1$: Take the all-1 column and one zero from each of the remaining columns.  This certificate is not always unique, however, as there can be multiple zeroes in a column and any of them can be chosen in a certificate.  Indeed, it is impossible to \emph{partition} the set of all positive inputs into subcubes of small length. 
%\todo{Alex: I still don't like this.  These ``subcubes'' fall out of the blue, and what is the length of a subcube is totally not clear.  This part can be said better.}

\goos~\etal added a surprisingly simple ingredient that solves this problem: pointers to cells in $M$. 
One can specify which zero to take from each column by requiring that there is a path of pointers that starts in the all-1 column and visits exactly one zero in all other columns, see \rf(fig:1cert-path).
Thus, the set of positive inputs breaks apart into a disjoint union of subcubes of small length.  
Since the pointers provide great flexibility in the positioning of zeroes, this function is still hard for a deterministic algorithm.
\mycommand{scaling}{0.8}

\begin{figure}[htbp]
\centering

\begin{tikzpicture}[semithick, ->, scale=\scaling]

\tikzstyle{cell}=[rectangle,fill=black!10,draw=black,very thick,minimum size=\scaling cm,inner sep=0pt]

\coordinate (pointer) at (+0.3,+0.3);

\draw[step=1cm,gray,very thin,-] (0, 0) grid (8,8);

%1-column
\foreach \y in {0, ..., 7} {
    \node[cell]    (c\y) at (3+0.5, \y+0.5) {$1$};
}
\foreach \y in {0, ...,  4, 6, 7} {
    \node[]    () at ([shift={(pointer)}]c\y) {\scriptsize $\bot$};
}

%zeroes
\foreach \x/\y [count=\i] in 
    {5/3, 6/6, 0/6, 1/5, 2/3, 4/1, 7/2}
{
    \node[cell]    (a\i) at (\x+0.5, \y+0.5) {$0$};
}

\draw (c5) ++(pointer) .. controls +(0:1) and +(90:1) .. (a1);
\draw (a1) ++(pointer) .. controls +(30:1) and +(-75:1) ..  (a2);
\draw (a2) ++(pointer) .. controls +(150:2) and +(30:1) .. (a3);
\draw (a3) ++(pointer) .. controls +(-15:0.5) and +(90:0.7) .. (a4);
\draw (a4) ++(pointer) .. controls +(-15:0.5) and +(90:1)  .. (a5);
\draw (a5) ++(pointer) .. controls +(30:2.5) and +(90:1) .. (a6);
\draw (a6) ++(pointer) .. controls +(60:0.75) and +(180:1) .. (a7);

\end{tikzpicture}

\caption{An example of a $1$-certificate for  the \goosfunction. 
    The center of a cell $x_{i,j}$ shows $\mathrm{val}(x_{i,j})$ and the top right corner shows $\mathrm{point}(x_{i,j})$}
\label{fig:1cert-path}
\end{figure}
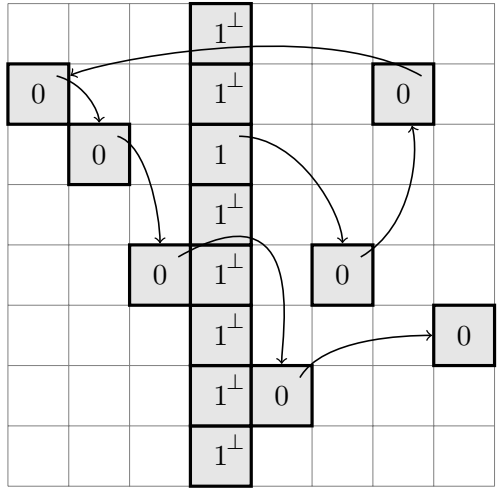

\mycommand{point}{\mathop{\mathrm{point}}}
\mycommand{val}{\mathop{\mathrm{val}}}
Formally, the definition of the \goosfunction is as follows.
Let $n$ and $m$ be positive integers, and $M=[n]\times [m]$ be a grid with $n$ rows and $m$ columns.  Let $\tM = M\cup\{\bot\}$.  Elements in $\tM$ are considered as pointers to the cells of $M$, where $\bot$ stands for the null pointer.

The function $g_{n,m}\colon (\bool \times \tM)^M \rightarrow \bool$ is defined as follows.  We think of each tuple $v = (b,p) \in \bool \times \tM$ in the following way.  The element $b \in \bool$ of the pair is the \emph{value} and the second element $p \in \tM$ is the \emph{pointer}.  We will use the notation 
$\val(v)=b$ and $\point(v)=p$.

Although $g_{n,m}$ is not a boolean function, it can be converted into an associated boolean function by encoding the 
elements of the input alphabet $\Sigma= \bool \times \tM$ using $\ceil[\log|\Sigma|]$ bits. 

An input $(x_{i,j})_{(i,j)\in M}$ evaluates to $1$ if and only if the following three conditions are satisfied (see \rf(fig:1cert-path) for an illustration):
\enumstart
  \item  There is exactly one column $b$ such that $\val(x_{i,b})=1$ for all $i \in [n]$.  We call this the \emph{marked column}. 
  \item In the marked column, there exists a unique cell $a$ such that $x_{a} \neq (1,\bot)$.  We call $a$ the \emph{special element}.
  \item  For the special element $a$, by following the pointers inductively defined as $p_1 = \point(x_a)$ and $p_{s+1} = \point(x_{p_s})$ for $s =1, \ldots, m-2$ we visit every column except the marked column, and $\val(x_{p_s})=0$ for each $s=1, \ldots, m-1$.
	We call $p_1,\dots,p_{m-1}$ the \emph{highlighted zeroes}.
\enumend
For each positive input $x$, the all-1 column satisfying items~(1) and~(2) of the definition, and the highlighted zeroes from~(3) give a unique minimal 1-certificate of $x$.
Thus, the unambiguous $1$-certificate complexity of this function is $n+m-1$.  \goos~\etal showed that this function has deterministic query complexity $mn$, giving a quadratic separation between the two when $n=m$.
%As the unambiquous nondeterministic query complexity is an upper bound on the polynomial degree, this also implies a quadratic separation between deterministic query complexity and degree, improving the previous best separation by Kushilevitz \cite{NW95}.

\subsection{Our technique and pointer functions}
As described in the previous section,
\goos~\etal showed how pointers can make certificates unambiguous without substantially 
increasing their size.  This technique turns out to be quite powerful for other applications as well.

Using the \goosfunction, it is already possible to give a larger separation between randomized and deterministic query complexity than previously known.
For instance, Mukhopadhyay and Sanyal~\cite{mukhopadhyay:goos}, independently from our work, obtained separations $R(f) = \tO(\sqrt{D(f)})$ and $R_0(f) = \tO(D(f)^{3/4})$.
However, these algorithms are rather complicated, and it is not known whether this function can realize an optimal separation between $R_0(f)$ and $D(f)$.

We instead modify the \goosfunction in various ways.  For the separation between $R_0$ and $D$ the key new idea we 
add is the use of \emph{back pointers}; for the separation between $Q$ and $D$, in addition to back pointers, we
further replace the path of zeroes in the \goosfunction with a \emph{balanced tree} whose leaves are the highlighted zeroes; finally, we 
consider a modification of the \goosfunction where there are \emph{multiple} marked columns for the separation 
between $Q$ and $R_0$. 
We now describe our modifications in more detail. 

%As described in the previous section,
%\goos~\etal showed how pointers can make certificates unambiguous without substantially 
%increasing their size.  
%This technique turns out to be quite powerful for other applications as well.
%In fact, all results in this paper stem from a systematic application of this technique, in combination with several new ideas.

\paragraph{Back pointers}
%In order to obtain a nearly optimal separation between $R_0(f)$ and $D(f)$, we modify the \goosfunction.
%A 1-certificate for our function still consists of an all-1 column and one zero from each of the remaining columns, which can be accessed using pointers.  
%The most interesting part of the certificate are the zeroes.
%We equip the zeroes with additional information and impose additional constraints on them.  Different constraints make the function feasible for some models of computation and infeasible for others.  The all-1 column, as in the \goosfunction, serves the sole purpose to certify there is only one ``genuine'' collection of zeroes.  When the all-1 column is found, we can find the zeroes by following the pointers and verify that the constraints are met.
%
%The new pieces of information we add to the zeroes are \emph{back pointers}.
A back pointer points either to a cell in $M$ or to a column in $[m]$.  
For instance, in order to get a quadratic separation between $R_0(f)$ and $D(f)$, we require that each highlighted zero points back to the marked (all-1) column.
%$b$.  
It turns out that this function is still hard for a deterministic algorithm.  A randomized algorithm, on the other hand, 
can take advantage of the back pointers to quickly find the all-1 column, if it exists.  The algorithm begins by querying all elements in a column.
Let $Z$ be the set of zeroes in this column and $B(Z)$ be the set of columns pointed to by the back pointers in $Z$.  If the value of function is 1, $B(Z)$ must contain the marked column.  However, $B(Z)$ may also contain pointers to non-marked columns.  We estimate the number of zeroes in each column of $B(Z)$ by sampling.  If we find a zero in every column of $B(Z)$, then we can reject the input.  
On the other hand, we can tune the sampling so that if no zero is found in a column $c \in B(Z)$, 
then, with high probability, $c$ has at most $|Z|/2$ many zeroes.  We then move to this column $c$ and repeat the process.  
Even if $c$ is not marked, we have made progress by halving the number of zeroes, and, in a logarithmic number of repetitions, we either find the marked column or reject.

In this way, back pointers to the marked column from the highlighted zeroes make the function easy for an $R_0$ algorithm, but hard for a deterministic one.  Similarly, if we only require that at least half of the highlighted zeroes point back to the special element $a$ from condition~(2), the function becomes hard for $R_0$, but easy for $R_1$.

%Finally, if we do not include any back pointers, the function becomes hard even for an $R$ algorithm.  The last two separations, though, require another modification of the function that we are about to describe.

\paragraph{Making partial functions total}
From another vantage point, the pointer technique can essentially turn a partial function into a total one.  This is beneficial as it is easy to prove separations for a partial function.  
Let us describe our separation between Monte Carlo and Las Vegas query complexities as an indicative example.  

It is easy to provide a separation between $R(f)$ and $R_0(f)$ for partial functions.
For example, consider the following partial boolean function $f$ on $m$ variables.
For $x\in\bool^m$, the value of $f(x)$ is 1 if the Hamming weight $|x| \ge m/2$, and $f(x)=0$ if $|x|=0$.  Otherwise, the function is not defined.
The Monte Carlo query complexity of this function is $O(1)$, but its Las Vegas complexity is $m/2+1$,  since it takes that many queries to reject the all-0 string.

How can we obtain a total function with the same property that there are either exactly $0$ or at least $m/2$ marked elements?
We define a variant of the \goosfunction, where we require that, in a positive input, at least $m/2$ of the highlighted zeroes point back to the special element $a$ of condition~(2). 
Consider an auxiliary function $f$ on the columns of the grid $M$.  For a column $j\in[m]$, $f(j)=1$ if and only if the value of the original function is 1, and the highlighted zero in column $j$ points back to $a$.
Thus, by definition, either $f(j)=0$ for all $j$, or $f(j)=1$ for at least half of all $j \in [m]$.  
Given a column $j$, we can find $a$ by analyzing the back pointers contained in column $j$.  When $a$ is found, it is easy to test whether the value of the function is 1 and the highlighted zero in column $j$ points to $a$.  Moreover, this procedure can be made deterministic and uses only $\tO(n+m)$ queries.

Subsequent to our work, Ben-David \cite{ben-david:super-grover} devised a different way of converting a partial function into a total one, and applied it to the forrelation problem \cite{aaronson:forrelation} to give a function $f$ with $Q(f) = \tO(R(f)^{2/5})$, the first super-quadratic separation between these measures.

\paragraph{Use of a balanced tree}
Instead of a path through the highlighted zeroes as in condition~(3) of the \goosfunction, we use a balanced binary tree with the zeroes being the leaves of the tree.  This serves at least three purposes.

First, this allows for even greater flexibility in placing the zeroes.  As they are the leaves of the tree, they are not required to point to other nodes.  This helps in proving Las Vegas lower bounds.

Second, the tree allows ``random access'' to the highlighted zeroes.
This is especially helpful for quantum algorithms.  After the algorithm finds the marked column, it should check the highlighted zeroes.  The last element of the path can be only accessed in $m$ queries.  But if we arrange the zeroes in a binary tree, each zero can be accessed in only a logarithmic number of queries, hence, they can be tested in $\tO(\sqrt{m})$ queries using Grover's search.

Finally, it can make the function hard even for a Monte Carlo algorithm (if no back pointers are present).
In the original \goosfunction, when a randomized algorithm finds a highlighted zero, it can follow the path starting from that cell.  As the zeroes are arranged in a path, the algorithm can thus eliminate half of the potential marked columns on average.  This fact is exploited by the algorithm of Mukhopadhyay and Sanyal~\cite{mukhopadhyay:goos}.

Similarly, when an algorithm finds a node of the tree it can also explore the corresponding subtree.  The difference is that the expected size of a subtree rooted in a node of the tree is only logarithmic.  Thus, even if the algorithm finds this node, it does not learn much, and we are able to prove an $\Omega(nm/\log m)$ lower bound for a Monte Carlo algorithm.

In principle, all of the above problems can be solved by adding direct pointers from the special element ($a$ in condition~(2)) to a zero in each non-marked column (that is, using an $(m-1)$-ary tree of depth 1 instead of a binary tree of depth $O(\log m)$).
The problem with this solution is that the size of the alphabet becomes exponential, rendering this construction useless for boolean functions.

\paragraph{Choice of separating functions}
We have outlined above three ingredients that can be added to the original \goosfunction: using various back pointers, identifying unmarked columns by a tree of pointers, and increasing the number of marked columns.
These ingredients can be added in various combinations to produce different effects.
In order to reduce the number of functions introduced, in this paper we stick to three variations:
\itemstart
\item a function $f_{n,m}$ with back pointers to the marked column from each of the highlighted zeroes.
\item a function $g_{n,m}$ with back pointers to the special element from half of the highlighted zeroes, and
\item a function $h_{k,n,m}$ with $k$ marked columns and no back pointers.
\itemend 
All three functions use a balanced binary tree.

This does not mean that a given separation cannot be proven with a different combination of ingredients.  For example, as outlined above, the $R_0$ vs $D$ separation can be proven for the original \goosfunction by equipping each highlighted zero with a back pointer to the marked column and not using the binary tree.

%%%%%%%%%%%%%%%%%%
\section{Preliminaries}
\label{sec:prelim}
%%%%%%%%%%%%%%%%%%
We let $[n]=\{1,2, \ldots, n\}$.  We use $f(n)=\tO(g(n))$ to mean that there exists constants $c,k$ and an integer $N$ 
such that $|f(n)| \le c |g(n)| \log^k(n)$ for all $n > N$.   
%\todo{removed $\widetilde \Omega$, no longer used.}
%Similarly we let $f(n)=\widetilde \Omega(g(n))$ to mean that 
%there exists constants $c,k$ and an integer $N$ such that $|f(n)| \ge c \tfrac{|g(n)|}{\log(n)^k}$ for all $n > N$.

In the remaining part of this section, we define the notion of query complexity for various models of computation.
For more detail on this topic, the reader may refer to the survey~\cite{buhrman:querySurvey}.
Relations between various models are depicted in \rf(fig:complexities).

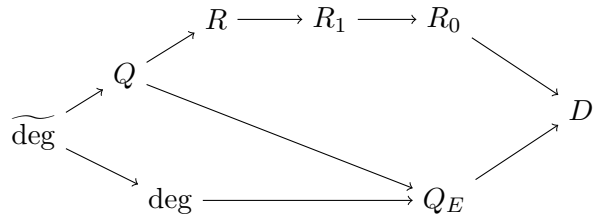
\begin{figure}[tbhp]
    \centering
    \begin{tikzpicture}[->, scale=0.6]
    \node     (adeg) at   (0, 1.5) {$\widetilde{\deg}$};
    \node     (deg)  at   (3,   0)  {$\deg$};
    \node     (q2)   at   (2, 2.75)  {$Q$};
    \node     (r2)   at   (4,   4)  {$R$};
    \node     (r1)   at   (6.5,   4)  {$R_1$};
    \node     (r0)   at   (9,   4)  {$R_0$};
    \node     (qe)   at   (9,   0)  {$Q_E$};
    \node     (d)    at   (12,   2)  {$D$};
    
    \draw (adeg) -- (q2);
    \draw (adeg) -- (deg);
    \draw (q2) -- (r2);
    \draw (r2) -- (r1);
    \draw (r1) -- (r0);
    \draw (r0) -- (d);
    \draw (deg) -- (qe);
    \draw (qe) -- (d);
    \draw (q2) -- (qe);
    \end{tikzpicture}
    \caption{Relations between various complexities.  An arrow means that complexity on the left is at most the complexity on the right.}
    \label{fig:complexities}
\end{figure}

\paragraph{Deterministic query complexity}
Let $\Sigma$ be a finite set.  A decision tree $T$ on $n$ variables and the input alphabet $\Sigma$ is a rooted tree, where
\itemstart
  \item internal nodes are labeled by elements of $[n]$;
  \item every internal node $v$ has degree $|\Sigma|$ and there is a bijection between the edges from $v$ to its  children and the elements of $\Sigma$;
  \item leaves are labeled from $\bool$.
\itemend
The output of the decision tree $T$ on input $x \in \Sigma^n$, denoted $T(x)$, is determined as follows.  
Start at the root.  If this is a leaf, then output 
its label.  Otherwise, if the label of the root is $i \in [n]$, then
follow the edge labeled by 
$x_i$ (this is called a \emph{query}) and recursively evaluate the corresponding subtree.  We say that $T$ computes the function 
$f\colon\Sigma^n \rightarrow \bool$ if $T(x)=f(x)$ on every input $x$.  The cost of 
$T$ on input $x$, denoted 
$C(T,x)$, is the number of internal nodes visited by $T$ on $x$. 
The deterministic query complexity $D(f)$ of $f$ is the minimum over all decision trees $T$ computing $f$ of the maximum over all $x$ of $C(T,x)$.

\paragraph{Randomized query complexity}
We follow the definitions for randomized query complexity given in \cite{Yao77, nisan:bs}. 
A randomized decision tree $T_\mu$ is defined by a probability distribution $\mu$ over deterministic decision trees.  
On input $x$, a randomized decision tree first selects a deterministic decision tree $T$ according to $\mu$, and then 
outputs $T(x)$.  The expected cost of $T_\mu$ on input $x$ is the expectation of $C(T,x)$ when $T$ is 
picked according to $\mu$.  The worst-case expected cost of $T_\mu$ is the maximum over inputs $x$ 
of the expected cost of $T_\mu$ on input $x$.

There are three models of randomized decision trees that differ in the definition of ``computing'' a function $f$.
\itemstart
\item Zero-error (Las Vegas):  It is required that the algorithm gives the correct output with probability 1 for every 
input $x$, that is,
every deterministic decision tree $T$ in the support of $\mu$ computes $f$.  
%Query complexity denoted by $R_0$.
\item One-sided error:  It is required that negative inputs are rejected with probability 1, and positive inputs are accepted with probability at least $1/2$.
%Query complexity denoted by $R_1$.
\item Two-sided error (Monte Carlo):  It is required that the algorithm gives the correct output with probability at least $9/10$ for every input $x$. 
%Query complexity denoted by $R$. 
\itemend

The error probability in the one-sided and two-sided cases can be reduced to $\eps$ by repeating the algorithm $O(\log \frac1\eps)$ times.

We define randomized query complexities $R_0(f), R_1(f)$, and $R(f)$ as the minimum worst-case expected cost of a randomized decision tree to compute $f$ in the zero, one-sided, and bounded-error sense, respectively.

\paragraph{Distributional query complexity}
A common way to show lower bounds on randomized complexity, and the way we will do it in this paper, is to consider 
distributional complexity \cite{Yao77}.  
The cost of a deterministic decision tree $T$ with respect to a distribution $\nu$, denoted $C(T,\nu)$, 
is $\mathbb{E}_{x \leftarrow \nu}[C(T,x)]$.  The decision tree $T$ computes a function $f$ 
with distributional error at most $\delta$ if $\Pr_{x \leftarrow \nu}[T(x)=f(x)] \ge 1-\delta$.  Finally, the 
 $\delta$-error distributional complexity of $T$ with respect to $\nu$, denoted
$\Delta_{\delta,\nu}(f)$, is the minimum of $C(T, \nu)$ over all $T$ that compute $f$ with distributional error at most 
$\delta$.

Yao has shown the following:
\begin{thm}[Yao \cite{Yao77}]
\label{thm:yao}
For any distribution $\nu$ and a function $f$, $R_0(f) \ge \Delta_{0,\nu}(f)$ and $R(f) \ge \tfrac{1}{2}\Delta_{2/10,\nu}(f)$.  
\end{thm}
In general, Yao shows that the $\delta$-error randomized complexity of a function $f$ is at least 
$\tfrac{1}{2}\Delta_{2\delta,\nu}(f)$, for any distribution $\nu$.  We obtain the constant $\tfrac{2}{10}$ on the right hand 
side as we have defined $R(f)$ for algorithms that err with probability at most $\tfrac{1}{10}$.

%We say that a zero-error (Las Vegas)
%randomized protocol $T_\mu$ computes $f$ if each deterministic 
%decision tree in the support of $\mu$ computes $f$.  The expected cost of $T_\mu$ on input $x$ is the expected number of 
%queries made on input $x$, according to the distribution $\mu$.  The zero-error complexity $R_0(f)$ of $f$ is the 
%minimum, over all zero-error protocols $T_\mu$ computing $f$, of the maximum expected cost of an input on $T_\mu$.
%
%We will also define a two-sided bounded-error notion of randomized complexity, also known as Monte Carlo complexity.  In 
%this model, we say that a randomized protocol $T_\mu$ computes $f$ with error $1/3$ if $T_\mu$ evaluates to $f(x)$ with 
%probability at least $2/3$ for every input $x$.  The cost of the protocol $T_\mu$ is the maximum depth of a decision tree in 
%the support of $\mu$.  The two-sided error complexity of $f$, denoted $R(f)$ is the minimum cost of a randomized 
%protocol that computes $f$ with error $1/3$.
%
%Finally, we define a one-sided error notion of randomized complexity.  A protocol $T_\mu$ computes a function $f$ 
%with one-sided error if the protocol always outputs $0$ if $f(x)=0$, and outputs $1$ if $f(x)=1$ with probability at least 
%$2/3$.  The one-sided error complexity $R_1(f)$ is defined as the minimum cost of a randomized protocol that 
%computes $f$ with one-sided error.  

\paragraph{Quantum query complexity}
The main novelty in a quantum query algorithm is that queries can be made in superposition.  For this exposition 
we assume $\Sigma = [|\Sigma|]$ (this identification can be made in an arbitrary way).  The memory of a 
quantum query algorithm contains two registers, the query register $H_Q$ which holds two integers $j \in [n]$ and 
$p \in \Sigma$ and the workspace $H_W$ which holds an arbitrary value.  
A query on input $x$  is encoded as a unitary operation $O_x$ in the following way.  On input $x$ and an 
arbitrary basis state $\ket |j,p>\ket |w> \in H_Q \otimes H_W$,
\[
O_x \ket |j,p> \ket |w> = \ket |j, p+x_j \bmod |\Sigma|>\ket |w> \, .
\]
A quantum query algorithm begins in the initial state $\ket |0,0> \ket |0>$ and on input $x$ proceeds by interleaving 
arbitrary unitary operations independent of $x$ and the operations $O_x$.  The cost of the algorithm is the number of 
applications of $O_x$.  The outcome of the algorithm is determined by a two-outcome measurement, specified by a complete set of projectors 
$\{\Pi_0, \Pi_1\}$.  If $\ket|\Psi_x>$ is the final state of the algorithm on input $x$, the probability that the algorithm 
outputs $1$ is $\|\Pi_1 |\Psi_x\rangle\|^2$.  
The exact quantum query complexity of the function $f$, denoted $Q_E(f)$, is the minimum 
cost of a quantum query algorithm that outputs $f(x)$ with probability $1$ for every input $x$.  The bounded-error quantum 
query complexity of the function $f$, denoted $Q(f)$, is the minimum cost of a quantum query algorithm that outputs $f(x)$ with probability at least $9/10$ for every input $x$.

We will describe our quantum algorithms as classical algorithms which use the following well-known quantum algorithms
as subroutines.  Let $O_x$ be a quantum oracle encoding a string $x\in\cube$.
\begin{itemize}
\item Grover's search \cite{grover:search, brassard:amplification}: Assume it is known that $|x| \ge t$.  There
is a quantum algorithm using $O(\sqrt{n/t})$ queries to $O_x$ that finds an $i$ such that $x_i=1$ with probability at least $9/10$.

\item Exact Grover's search \cite{brassard:amplification}. Assume it is known that $|x|=t$. There 
is a quantum algorithm using $O(\sqrt{n/t})$ queries to $O_x$ that finds an $i$ such that $x_i=1$ with certainty.  The case $t=n/2$ is essentially the Deutsch-Jozsa problem~\cite{deutsch:jozsa}. 

\item Approximate counting \cite{brassard:counting}: Let $t = |x|$.  There is a quantum algorithm making $O(\sqrt{n})$
queries to $O_x$ that 
outputs a number $\tilde t$ satisfying $| \tilde t- t | \le \tfrac{t}{10}$ with probability at least $9/10$. 
\item Amplitude amplification \cite{brassard:amplification}:
Assume a quantum algorithm $\cA$ prepares a state $\ket|\psi> = \alpha_0 \ket|0>\ket|\psi_0> + \alpha_1 \ket|1>\ket|\psi_1>$, where $\psi$, $\psi_0$ and $\psi_1$ are unit vectors, and $\alpha_0$ and $\alpha_1$ are real numbers.
Thus, the success probability of $\cA$, i.e., probability of obtaining 1 in the first register after measuring $\ket|\psi>$, is $\alpha_1^2$.

Assume a lower bound $p$ is known on $\alpha_1^2$.
There exists a quantum algorithm that makes $O(1/\sqrt{p})$ calls to $\cA$, and either fails, or generates the state $\ket|1>\ket|\psi_1>$.  The success probability of the algorithm is at least $9/10$.
 \end{itemize}

In all of these quantum subroutines, the error probability can be reduced to $\eps$ by repeating the algorithm $O(\log \frac1\eps)$ times.
 
 %\begin{thm}[Quantum Amplitude Amplification~\cite{brassard:amplification}]
%Let $\cA$ be a quantum algorithm that prepares a state $\ket|\psi> = \alpha_0 \ket|0>\ket|\psi_0> + \alpha_1 \ket|1>\ket|\psi_1>$, where $\psi$, $\psi_0$ and $\psi_1$ are unit vectors.
%There exists a quantum algorithm that calls $\cA$ $O(1/|\alpha_1|)$ times and generates the state $\ket|1>\ket|\psi_1>$ with constant success probability.
%\end{thm}
%
%\begin{thm}[Quantum Counting~\cite{brassard:counting}]
%Let $O_x$ be a quantum oracle encoding a sting $x\in\cube$.
%There exists a quantum procedure that estimates the Hamming weight of $x$ with relative precision $1/10$ using $O(\sqrt{n})$ queries to $O_x$.  The algorithm has error probability at most $1/3$.
%\end{thm}

%In our algorithms, we use only well-known tools.  Let us only define Grover's search~\cite{grover:search} here as our most commonly used tool in quantum algorithms.  
%Let $O_x$ be a quantum oracle encoding a sting $x\in\cube$, and $1\le k<n$ be a fixed integer.
%If it is known that $x$ contains at least $k$ elements equal to 1, Grover's search finds one of them with constant error probability using $O(\sqrt{n/k})$ queries to $O_x$.  
%If $x$ contains exactly $k$ elements equal to 1, Grover's search finds one with probability 1.
%\todo{Reference for exact Grover--?}

\paragraph{Polynomial degree}
Every boolean function $f\colon \bool^n \rightarrow \bool$ has a unique expansion as a multilinear polynomial 
$p = \sum_{S \subseteq [n]} \alpha_s \prod_{i \in S} x_i$.  The 
\emph{degree} of $f$, denoted $\deg(f)$, is the size of a largest monomial  $x_S$ in $p$ with 
nonzero coefficient 
$\alpha_S$.  The \emph{approximate degree} of $f$, denoted $\tdeg(f)$, is 
\[
\tdeg(f) = \min\sfig{ \deg(g) \midA |g(x) - f(x)| \le \tfrac{1}{10} \text{ for all } x \in \bool^n } \enspace.
\]
For any quantum algorithm that uses $T$ queries to the quantum oracle $O_x$, its acceptance probability is a polynomial of degree at most $2T$ \cite{beals:pol}.
Therefore, $\deg(f)\leq 2Q_E(f)$ and $\tdeg(f)\leq 2Q(f)$.

\paragraph{Certificate complexity}
A \emph{partial assignment} in $\Sigma^n$ is a string in $a \in (\Sigma \cup \{\star\})^n$.  The length of a partial 
assignment is the number of non-star values.  A string $x \in \Sigma^n$ is \emph{consistent} with an assignment $a$ if 
$x_i=a_i$ whenever $a_i \ne \star$.  Every partial assignment defines a \emph{subcube}, which is the set of all 
strings consistent with that assignment.  For every subcube there is a unique partial assignment that defines it, and 
we define the length of a subcube as the length of this assignment. 

%A \emph{subcube} is a set $S \subseteq \Sigma^n$ of all strings consistent with 
%some partial assignment.  The length of a subcube is the length of its 
%associated partial assignment.  
%\todo{removed $b$-subcube---not needed for defining one-sided versions.}

For $b\in\bool$, a \emph{$b$-certificate} for a function $f\colon\Sigma^n\to\bool$ is a partial assignment such that the value of $f$ is $b$ for all inputs in the associated subcube. 
The \emph{$b$-certificate complexity} of $f$ is the smallest number $k$ such that the set $f^{-1}(b)$ can be 
written as a union of subcubes of length at most $k$.
The \emph{unambiguous $b$-certificate complexity} of $f$ is the smallest number $k$ such that the set $f^{-1}(b)$ can be 
written as a disjoint union of subcubes of length at most $k$.

%An \emph{assignment} in $\Sigma^n$ is the set of all inputs $x=(x_j)\in\Sigma^n$ given by some conditions $x_{j_1} = v_1,\;\dots,\; x_{j_k} = v_k$ with $j_s\in[n]$ and $v_s\in\Sigma$ for all $s\in[k]$.  The number $k$ is called the \emph{length} of the assignment.
%
%For $b\in\bool$, a \emph{$b$-certificate} for a function $f\colon\Sigma^n\to\bool$ is an assignment such that the value of $f$ is $b$ for all inputs in the assignment.
%The \emph{$b$-certificate complexity} of $f$ is the smallest number $k$ such that the set $f^{-1}(b)$ can be covered by 1-certificates of length $k$.
%The \emph{unambiguous $b$-certificate complexity} of $f$ is the smallest number $k$ such that the set $f^{-1}(b)$ can be partitioned into 1-certificates of length $k$.

\paragraph{Booleanizing a function}
While we define functions over a nonboolean alphabet $\Sigma$, it is more typical in query complexity to discuss boolean 
functions.  Fix a surjection $b\colon \bool^{\ceil [\log|\Sigma|]} \rightarrow \Sigma$.  For a function 
$f\colon \Sigma^n \rightarrow \bool$, we define the associated boolean function 
$\tilde f\colon \bool^{n \ceil[\log|\Sigma|]} \rightarrow \bool$ by 
$\tilde f(x) = f(b(x))$.  A lower bound on $f$ in the model where a query returns an element of $\Sigma$ 
will also apply to $\tilde f$ in the model where a query returns a boolean value.  Also, if $f$ can be computed with $t$ 
queries then we can convert this into an algorithm for computing $\tilde f$ with $t \ceil [\log |\Sigma|]$ queries by querying 
all the bits of the desired element.   We will state our theorems for nonboolean functions where a query returns an 
element of $\Sigma$ and the alphabet size $|\Sigma|$ will always be polynomial in the input length.
By the remarks above, such separations can be converted into separations for the associated 
boolean function with a logarithmic loss.

%\begin{prp}
%Assume there is a quantum algorithm that evaluates a function $f$ with error probability at most $1/3$ in $T$ queries.  Then, there exists a quantum algorithm that evaluates $f$ with error probability at most $1/\eps$ using $O(T\log(1/\eps))$ queries.
%\end{prp}

%\begin{thm}[Quantum Amplitude Amplification~\cite{brassard:amplification}]
%Let $\cA$ be a quantum algorithm that prepares a state $\ket|\psi> = \alpha_0 \ket|0>\ket|\psi_0> + \alpha_1 \ket|1>\ket|\psi_1>$, where $\psi$, $\psi_0$ and $\psi_1$ are unit vectors.
%There exists a quantum algorithm that calls $\cA$ $O(1/|\alpha_1|)$ times and generates the state $\ket|1>\ket|\psi_1>$ with constant success probability.
%\end{thm}
%
%\begin{thm}[Quantum Counting~\cite{brassard:counting}]
%Let $O_x$ be a quantum oracle encoding a sting $x\in\cube$.
%There exists a quantum procedure that estimates the Hamming weight of $x$ with relative precision $1/10$ using $O(\sqrt{n})$ queries to $O_x$.  The algorithm has error probability at most $1/3$.
%\end{thm}

%%%%%%%%%%%%%%%%%%%%%%%%%%%%%
\section{Separations against deterministic complexity}
\label{sec:det}
%%%%%%%%%%%%%%%%%%%%%%%%%%%%%
\mycommand{bpoint}{\mathop{\mathrm{bpoint}}}
\mycommand{lpoint}{\mathop{\mathrm{lpoint}}}
\mycommand{rpoint}{\mathop{\mathrm{rpoint}}}

Let $n$, $m$, $M$, $\tM$ be as in the definition of the \goosfunction.
Let also $\tC = [m]\cup\{\bot\}$ be the set of pointers to the columns of $M$.
The input alphabet of our function is $\Sigma = \bool \times \tM \times \tM \times \tC$.
For $v\in \Sigma$, we call the elements of the quadruple the \emph{value}, the \emph{left pointer}, the \emph{right pointer} and the \emph{back pointer} of $v$, respectively.  We use notation $\val(v)$, $\lpoint(v)$, $\rpoint(v)$, and $\bpoint(v)$ for them in this order.

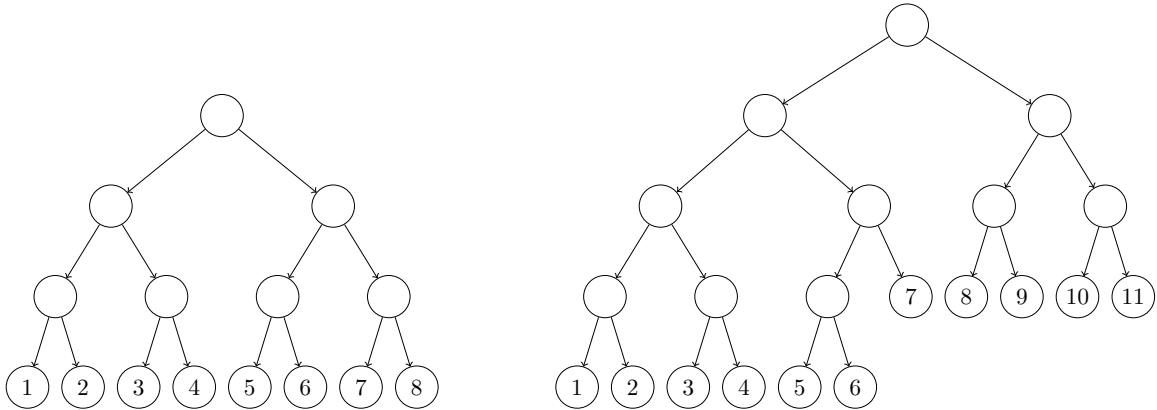
\begin{figure}[thbp]
\[
    \begin{tikzpicture}[->, scale=0.8,
    level distance = 1.5cm,
    sibling distance = 0.2cm,
    every tree node/.style={vert}]
    \tikzstyle{vert}=[circle,draw=black,minimum size=0.7cm,inner sep=0pt]
\tikzset{edge from parent/.style=
    {draw, edge from parent path={(\tikzparentnode) -- (\tikzchildnode)}}}

    \Tree
        [.{ } 
            [.{ }
                [.{ } {1} {2} ]
                [.{ } {3} {4} ]
            ]
            [.{ }
                [.{ } {5} {6} ]
                [.{ } {7} {8} ]
            ]
        ]
    \end{tikzpicture}
\qquad\qquad		
    \begin{tikzpicture}[->, scale=.8,
    level distance = 1.5cm,
    sibling distance = 0.2cm,
    every tree node/.style={vert}]
    \tikzstyle{vert}=[circle,draw=black,minimum size=0.7cm,inner sep=0pt]
    \tikzset{edge from parent/.style=
        {draw, edge from parent path={(\tikzparentnode) -- (\tikzchildnode)}}}
    \Tree
    [.{ } 
        [.{ }
            [.{ }
                [.{ } {$1$} {$2$} ]
                [.{ } {$3$} {$4$} ]
            ]
            [.{ }
                [.{ } {$5$} {$6$} ]
                {$7$}
            ]
        ]
        [.{ }
            [.{ }
                $8$
                $9$
            ]
            [.{ }
                $10$
                $11$
            ]
        ]
    ]
    \end{tikzpicture}
\]
    \caption{A completely balanced tree on 8 leaves, and a balanced tree on 11 leaves.}
    \label{fig:tree}
\end{figure}

\mytxtcommand{lft}{`left'}
\mytxtcommand{rht}{`right'}
Let $T$ be a fixed balanced oriented binary tree with $m$ leaves and $m-1$ internal vertices.
For instance, we can make the following canonical choice.
If $m=2^k$ is a power of two, we use the completely balanced binary tree on $m$ leaves as depicted in \rf(fig:tree) on the left.  Each leaf is at distance $k$ from the root.
Otherwise, assume $2^k<m<2^{k+1}$.  Take the completely balanced tree on $2^k$ leaves, and add a pair of children to each of its $m-2^k$ leftmost leaves.  An example is in \rf(fig:tree) on the right.

We have the following labels in $T$.
The outgoing arcs from each node are labeled by \lft and \rht.
The leaves of the tree are labeled by the elements of $[m]$ from left to right, 
with each label used exactly once.
For each leaf $j\in[m]$ of the tree, the path from the root to the leaf defines a sequence of \lft and \rht of length $O(\log m)$, which we denote $T(j)$.

The function $f_{n,m}\colon \Sigma^M\to\bool$ is defined as follows.
For an input $x = (x_{i,j})$, we have $f_{n,m}(x)=1$ if and only if the following 
conditions are satisfied (for an illustration refer to \rf(fig:1cert-tree-bp)):
\enumstart
  \item There is exactly one column $b\in[m]$ such that $\val(x_{i,b})=1$ for all $i \in [n]$.  We refer to it as the \emph{marked column}.
  \item In the marked column, there exists a unique cell $a$ such that $x_{a} \neq (1,\bot,\bot,\bot)$.  We call $a$ the \emph{special element}.
	\item\label{3} For each non-marked column $j\in [m]\setminus\{b\}$, let $\ell_j$ be the end of the path which starts at the special element $a$ and follows the pointers $\lpoint$ and $\rpoint$ as specified by the sequence $T(j)$.
	We require that $\ell_j$ exists (no pointer on the path is $\bot$), $\ell_j$ is in the $j$th column, and $\val\sA[x_{\ell_j}] = 0$.
	We call $\ell_j$ the \emph{leaves of the tree}. 
	\item Finally, for each non-marked column $j\in[m]\setminus\{b\}$, we require that $\bpoint\sA[x_{\ell_j}] = b$.
\enumend

\begin{figure}[tbph]
    \centering
    
    \begin{tikzpicture}[semithick, ->, scale=\scaling]
    
    \clip(-1,-1.25) rectangle (9,8); %otherwise the curved backpointers make canvas too large
    
    \tikzstyle{cell}=[rectangle,fill=black!10,draw=black,very thick,minimum size=\scaling cm,inner sep=0pt]
    
    \coordinate (lp) at  (-0.3, 0.0);
    \coordinate (rp) at  (+0.3, 0.0);
    \coordinate (bp) at  ( 0.0,-0.3);
    \coordinate (val) at ( 0.0, 0.10);

    \coordinate (botc) at (3.5,0);
                
    \draw[step=1cm,gray,very thin,-] (0, 0) grid (8,8);

    %1-column
    \foreach \y in {0, ..., 7} {
        \node[cell]    (c\y) at (3+0.5, \y+0.5) { };
        \node     ()  at ([shift={(val)}]c\y) {$1$};
    }
    \foreach \y in {0, ...,  5, 7} {
        \node    () at ([shift={(lp)}]c\y) {\scriptsize $\bot$};
        \node    () at ([shift={(rp)}]c\y) {\scriptsize $\bot$};
        \node    () at ([shift={(bp)}]c\y) {\scriptsize $\bot$};
    }
    \coordinate (spec) at (c6);
    \node    () at ([shift={(bp)}]spec) {\scriptsize $\bot$};

    %zeroes
    \foreach \x/\y [count=\i] in 
    {0/1, 1/0, 2/1, 4/1, 5/1, 6/2, 7/3}
    {
        \node[cell]    (a\x) at (\x+0.5, \y+0.5) {$0$};
%        \node          ()  at ([shift={(val)}]a\x) {$0$};
    }

    %tree inner nodes
    \foreach \x/\y [count=\i] in 
    {1/4, 5/5, 0/3, 2/3, 4/3, 6/4}
    {
        \node[cell]    (t\i) at (\x+0.5, \y+0.5) { };
    }

    %tree pointers
    \draw (spec) ++(lp) .. controls +(-180:1.25)  and +(90:1.0) .. (t1);
    \draw (spec) ++(rp) .. controls +( 15:0.5)  and +(90:1) .. (t2);
    \draw (t1) ++(lp) .. controls +(180:0.4)   and +( 90:0.75) .. (t3);
    \draw (t1) ++(rp) .. controls +(   0:0.4)  and +(90:0.75) .. (t4);
    \draw (t2) ++(lp) .. controls +(180:0.5)   and +( 90:1) .. (t5);
    \draw (t2) ++(rp) .. controls +( 0:0.4)    and +( 90:0.75) .. (t6);
    
    \draw (t3) ++(lp) .. controls +(-150:0.5)   and +( 120:1.25) .. (a0);
    \draw (t3) ++(rp) .. controls +(-15:0.5)   and +( 90:1) .. (a1);
    \draw (t4) ++(lp) .. controls +(-150:0.5)   and +(120:1.25) .. (a2);

    \draw (t5) ++(lp) .. controls +(-150:0.5)   and +(120:1.25) .. (a4);
    \draw (t5) ++(rp) .. controls +(-30:0.75)   and +( 90:1.0) .. (a5);
    \draw (t6) ++(lp) .. controls +(-150:0.5)   and +(120:1.25) .. (a6);
    \draw (t6) ++(rp) .. controls +( 0:0.4)   and +( 90:0.75) .. (a7);

    %backpointers
    \begin{scope}[very thin]
    \draw (a0) ++(bp) .. controls +(-90:0.5)  and +(-135:3.0) .. ([xshift=-0.1cm]botc);
    \draw (a1) ++(bp) .. controls +(-75:0.5)  and +(-145:2.0) .. ([xshift=-0.2cm]botc);
    \draw (a2) ++(bp) .. controls +(-90:0.5)  and +(-145:1.0) .. ([xshift=-0.3cm]botc);
    \draw (a4) ++(bp) .. controls +(-90:0.5)  and +(-35:1.0)  .. ([xshift=+0.3cm]botc);
    \draw (a5) ++(bp) .. controls +(-90:0.5)  and +(-35:2.0)  .. ([xshift=+0.2cm]botc);
    \draw (a6) ++(bp) .. controls +(-90:0.5)  and +(-45:3.0)  .. ([xshift=+0.1cm]botc);
    \draw (a7) ++(bp) .. controls +(-90:0.5)  and +(-45:5.0)  .. ([xshift=+0.0cm]botc);
    \end{scope}

    \end{tikzpicture}
    
    \caption{An example of a $1$-certificate for the function $f_{8,8}$.  The tree $T$ is like in \rf(fig:tree) on the left.
        The center of a cell $x_{i,j}$ shows $\mathrm{val}(x_{i,j})$, the bottom of the cell shows $\mathrm{bpoint}(x_{i,j})$ and the bottom left and right sides show $\mathrm{lpoint}(x_{i,j})$ and $\mathrm{rpoint}(x_{i,j})$, respectively.
				Values and pointers that are not shown can be chosen arbitrarily.
				}
    \label{fig:1cert-tree-bp}
\end{figure}
\mytxtcommand{defoff}{the definition of $f_{n,m}$}

\begin{thm}
\label{thm:D-1}
If $n = 2m$ and $m$ is sufficiently large, the deterministic query complexity $D(f_{n,m}) \ge m^2$.
\end{thm}

\pfstart
We describe an adversary strategy that ensures that the value of the function is undetermined after $m^2$ queries, provided $m \ge 4$. 
Assume a deterministic query algorithm queries a cell $(i,j)$.
Let $k$ be the number of queried cells in column $j$, including the cell $(i,j)$.
If $k\le m$, the adversary replies with $(1,\bot,\bot,\bot)$.
Otherwise, the response is $(0,\bot,\bot,k-m)$.

Note that, after all the cells are queried in some column, it contains $m$ cells with $(1,\bot,\bot,\bot)$ and one cell with $(0,\bot,\bot,b)$ for each $b\in[m]$.

%On some moment in the algorithm, let $B$ be the set of columns, where at most $m$ cells were queried.
%The idea of the adversary strategy is summarized in the following claim.

\begin{clm}
If there is a column $b\in[m]$ with at most $m$ queried cells and there are at least $4m$ unqueried cells in total, then the function value is undetermined.
\end{clm}

\begin{proof}
First, the adversarial strategy is such that no all-one column can ever be constructed, hence, by answering all remaining queries with value $0$, the adversary can make the function evaluate to $0$.

Now we show that the function value can also be set to $1$.
For each column $j\ne b$, define $\ell_j$ as follows.
If column $j$ contains an unqueried cell $(i,j)$, let $\ell_j = (i,j)$, and assign the quadruple $(0,\bot,\bot,b)$ to this cell.
If all elements in column $j$ were queried, then, by the adversary strategy, it contains a cell with quadruple $(0,\bot,\bot,b)$.  Let $\ell_j$ be this cell.

Next, the queried cells in column $b$ only contain $(1, \bot, \bot, \bot)$.
Assign the quadruple $(1, \bot, \bot, \bot)$ to the remaining cells in column $b$ except for one special cell $a$.  Using the cell $a$ as the root construct a tree of pointers isomorphic to $T$ using as internal vertices some of the remaining unqueried cells, and such that the $j$th leaf is $\ell_j$.
Finally, assign the quadruple $(1, \bot, \bot, \bot)$ to every other cell.

To carry out the construction above, we need $m-2$ unqueried cells outside of column $b$ and the set of $\ell_j$s to place the internal vertices of the tree.  Since there are $2m$ cells in column $b$ and $m-1$ cells are used by $\ell_j$s, it suffices to have $4m$ unqueried cells to do this.
\end{proof}

It takes more than $m^2$ queries to ensure that each column contains more than $m$ queried cells. 
As $2m^2 - 4m \ge m^2$ when $m \ge 4$, we obtain the required lower bound.
\pfend

\begin{algorithm}[tbhp]
\caption{A Las Vegas randomized algorithm for the function $f_{n,m}$\label{alg:R0}}
\medskip
{\bf VerifyColumn($j$)} tests whether column $j$ is marked
\negmedskip
\enumstart
\item If column $j$ does not satisfy condition~(2) of \defoff, then reject.  Otherwise, let $a$ be the special element.
\item Following the left and right pointers from $a$ and querying the elements along the way, check that the tree rooted at $a$ satisfies conditions~(3) and~(4) of \defoff. 
 If it does, accept.  Otherwise, reject.
\enumend

{\bf TestColumn($c$, $k$)} always returns `True' if column $c$ has no zeroes.  
If it has more than $k/2$ zeroes, returns `False' with probability $\ge 1-1/(nm)^2$.
Returns anything in the intermediate cases.
\negmedskip
\enumstart
\item Query $O(\frac nk\log (nm))$ random elements from column $c$.  If no zero was found, return `True'.  Otherwise, return `False'.
\enumend

{\bf Main procedure of the algorithm}
\negmedskip
\enumstart
\item Let $j$ be an arbitrary column in $[m]$, and $k\gets n$.  

\item Repeat the following actions:
\negmedskip

\enumstart
\item Query all the elements of column $j$.  If all of them have value $1$, VerifyColumn($j$).
\item If column $j$ contains more than $k$ zeroes, then query all the elements of $M$ and output the value of the function.
\item Else, let $C$ be the set of nonnull back pointers stored in the zero elements of column $j$. 
For each $c\in C$, TestColumn($c$, $k$).  If `False' is obtained for all the columns, reject.  
Otherwise, let $j$ be any column with outcome `True'.
\item If $k=0$, reject.  Otherwise, let $k\gets \floor[k/2]$, and repeat the loop.
\enumend
\enumend
\end{algorithm}

\begin{thm}
\label{thm:R0-1}
The Las Vegas randomized complexity $R_0(f_{n, m}) = \tO(n + m)$.
\end{thm}

\pfstart
For the description, see \rf(alg:R0).
With each iteration of the loop in step~2, $k$ gets reduced by half until it becomes zero, hence, after $O(\log n)$ iterations of the loop, the algorithm terminates.

Let us check the correctness of the algorithm.
The algorithm only accepts from the procedure VerifyColumn which verifies the existence of a $1$-certificate.  Thus, the algorithm never accepts a negative input.

To see the algorithm always accepts a positive input, let the input $x$ be positive with marked column $b$.  
Consider one iteration of the loop in step~2.  If $j=b$, then the algorithm accepts in VerifyColumn($j$) on step~2(a).
Now assume $j\ne b$.
Then, column $j$ contains a zero with a back pointer to $b$, hence, the algorithm does not reject on step~2(c).
The algorithm also does not reject on step~2(d) since, when $k=0$, the condition in~2(b) applies.

%If $b$ is chosen in step~1, then VerifyColumn($b$) will accept.  Otherwise, as every column $j$ has a zero 
%element pointing back to $b$, we have $b \in C$.  The invariant $b \in C$ is maintained throughout the loop as 
%TestColumn($b,k$) always returns `True'.  If the loop ever exits through step~2(c) then it is always correct.  
%As only column $b$ passes step~1 of ProcessColumn, the algorithm is also correct if it exits there.
%It is only left to be seen that the algorithm does not reject in step~2(d).
%Consider the start of the loop when $k=1$.  If $b$ is the column chosen in step~2(c) then the algorithm will exit 
%correctly through VerifyColumn; otherwise, the algorithm will exit correctly through step~2(c), as $b$ is the only column 
%with less than $1/2$ many zeroes.

Let us now estimate the expected number of queries made by the algorithm.
Condition in step~2(b) is obviously not satisfied on the first iteration of the loop.
On a specific later iteration, the probability this condition is satisfied is at most $1/(nm)^2$ by our definition of 
TestColumn($c,k$).  
Since the loop is repeated $O(\log n)$ times, the contribution of step~2(b) to the complexity of the algorithm is $o(1)$.

If step~2(b) is not invoked, we have the following complexity estimates.
VerifyColumn uses $O(m)$ queries, and it is called at most once. 
Apart from VerifyColumn, Step~2(a) uses $n$ queries.
Since $|C|\le k$ on step~2(c), the number of queries in this step is $\tO(n)$.  
Since there is only a logarithmic number of iterations of the loop in step~2, the total number of queries is $\tO(n+m)$.
\pfend

%\mycommand{invlogsq}{{1/\log^2 n}}

\begin{algorithm}[tbhp]
\caption{A quantum algorithm for the function $f_{n,m}$\label{alg:Q}}
\medskip

{\bf VerifyColumn(j)} tests whether column $j$ is marked
\negmedskip
\enumstart
\item Use Grover's search to find an element $a$ in column $j$ with nonnull left or right pointer.  If no element found, reject.  If $\val(x_a)=0$, reject.
\item Use Grover's search to verify that all elements in column $j$ except $a$ are equal to $(1,\bot,\bot,\bot)$.  If not, reject.
\item Use Grover's search (over all $j\in[m]\backslash\{b\}$) to check that conditions~(3) and (4) of \defoff are satisfied.  If they are, accept.  Otherwise, reject.
\enumend

{\bf FindGoodBackPointer(j, k)} if column $j$ has $\le \frac{11}{10}k$ zeroes and one of them has a back pointer to an all-1 column, finds a column containing $\le k/2$ zeroes with probability $\ge 1/2k$.
\enumstart
\negmedskip
\item Use Grover's search to find a zero $v$ in column $j$.
\item If $\bpoint(x_v)=\bot$, return `False'.  Otherwise $c\gets\bpoint(x_v)$.
\item Execute Grover's search for a zero in column $c$, assuming there are $\ge k/2$ of them.
Return `False' if Grover's search finds a zero, and 'True' otherwise.
\enumend

{\bf Main procedure of the algorithm}
\negmedskip
\enumstart
\item Let $j$ be an arbitrary column in $[m]$.

\item Repeat the following actions.  If the loop does not finish after $10 \log n$ iterations, reject.
\negmedskip

\enumstart
\item Use quantum counting to estimate the number of zeroes in column $j$ with relative accuracy $1/10$.  Let $k$ be the estimate.  If $k=0$, VerifyColumn($j$).

\item Execute quantum amplitude amplification on the FindGoodBackPointer($j$, $k$) subroutine amplifying for the output `True' of the subroutine and assuming its success probability is at least $1/2k$.  Let $c$ be the corresponding value of the subroutine after amplification.

\item Set $j\gets c$.  Repeat the loop.
\enumend
\enumend
\end{algorithm}

\begin{thm}
\label{thm:Q-1}
The quantum query complexity $Q(f_{n,m}) = \tO(\sqrt{n}+\sqrt m)$.
\end{thm}

\pfstart
The algorithm, \rf(alg:Q), is a quantum counterpart of \rf(alg:R0).
We assume that every elementary quantum subroutine of the algorithm 
(e.g. Grover's search or quantum counting)
is repeated sufficient number of times to reduce its error probability to at most $1/(nm)^2$.  
This requires a logarithmic number of repetitions, which can be absorbed into the $\tO$ factor.
Since the algorithm makes less than $O(n+m)$ queries, we may further assume that all the 
elementary quantum subroutines are performed perfectly.

The analysis is similar to \rf(thm:R0-1).  
Again, the algorithm only accepts from VerifyColumn, which is called at most once.
The three steps of VerifyColumn correspond to the three conditions defining a $1$-input.  Any negative input violates one of these conditions, and thus will fail one of these tests.

Now suppose we have a positive input $x$ with marked column~$b$.  
%So, we may assume the input $x$ is positive.  Let~$b$ be its marked column.
In this case, each non-marked column contains a zero with a back pointer to the marked column $b$. 
We want to argue that the algorithm accepts $x$ with high probability.
The following claim is the cornerstone of our analysis.

\begin{clm}
\label{clm:halfing}
If the input $x$ is positive and column $j$ contains at most $\frac{11}{10}k$ zeroes, then step~2(b) of the algorithm finds a column $c$ containing at most $k/2$ zeroes with high probability.
\end{clm}

\pfstart
We first claim that FindGoodBackPointer($j$, $k$) returns `True' with probability at least $1/2k$.
Indeed, we assumed that the probability Grover's search on step 1 fails is negligible.  Thus, with high probability, before execution of step 2, $v$ is chosen uniformly at random from the at most $\frac{11}{10} k$ zeroes in column $j$.  One of these zeroes has a back pointer to the marked column $b$.  If it is chosen, step 3 returns `True' with certainty, which proves our first claim.

Thus, amplitude amplification in step~2(b) of the main procedure will generate the `True'-portion of the final state of the FindGoodBackPointer subroutine.  Again, since we assume that the error probability of Grover's search on step~3 of FindGoodBackPointer is negligible, we may assume this portion of the state only contains columns $c$ with at most $k/2$ zeroes.
\pfend

Consider the loop in step~2.  
We may assume quantum counting is correct in step~2(a).
If $j=b$, then VerifyColumn(j) is called in step 2(a), and the algorithm accepts with high probability.
So, consider the case $j\ne b$.
Column $j$ contains a zero, hence, with high probability, VerifyColumn is not executed on step~2(a).  Thus, by \rf(clm:halfing), the number of zeroes in column $j$ gets reduced by a factor of $1.1/2$.  Therefore, after $10 \log n$ iterations, the number of zeroes in column $j$ becomes zero, which means $j=b$, and the algorithm accepts with high probability.

We now estimate the complexity of the algorithm.  Steps~(1) and~(2) of VerifyColumn take $\tO(\sqrt{n})$ 
queries.  For step~(3) of VerifyColumn, we have to check that the tree of pointers rooted from $x_a$ satisfies conditions~(3) and~(4) from 
\defoff.  We can check the correctness of a single path from the root to a leaf with $O(\log m)$ (classical) queries.  
Since there are $m$ many paths, checking them all with Grover's search takes
$\tO(\sqrt{m})$ many queries.  Overall, VerifyColumn takes $\tO(\sqrt{n}+\sqrt{m})$ queries.

Grover's search in FindGoodBackPointer($j,k$) uses $\tO(\sqrt{n/k})$ queries.
Since the success probability of FindGoodBackPointer($j,k$) is at least $1/2k$, amplitude amplification
repeats FindGoodBackPointer($j,k$) $\tO(\sqrt{k})$ times and the complexity of step~2(b) is $\tO(\sqrt{n})$. 
Quantum counting in step~2(a) also uses $\tO(\sqrt{n})$ queries.

Since we run the main loop at most $O(\log n)$ many times, the total complexity of the algorithm is $\tO(\sqrt n + \sqrt m)$.
\pfend

\begin{cor}
\label{cor:R0-D}
There is a total boolean function $f$ with $R_0(f) = \tO(D(f)^{1/2})$ and $Q(f) = \tO(D(f)^{1/4})$.
\end{cor}

\begin{proof}
We first obtain these separations for a non-boolean function.
Take $f_{n,m}$ with $n = 2m$.  Then the zero-error randomized query complexity is $\tO(n)$ by \rf(thm:R0-1), the quantum  query complexity is $\tO(\sqrt{n})$ by \rf(thm:Q-1), and the deterministic query complexity is $\Omega(n^2)$ by \rf(thm:D-1).
Since the size of the alphabet $\Sigma$ is polynomial, this also gives the separations for the associated boolean 
function $\tilde f_{2m,m}$.
\end{proof}

%Using the Simulation Theorem from~\cite{goos:partitionNumber}, we immediately obtain the following consequence for communication complexity.
%\begin{cor}
%There exists a total Boolean function $f$ with $R_0^c(f) = \tO(D^c(f)^{1/2})$ and $Q^c(f) = \tO(D^c(f)^{1/4})$, where $R_0^c$, $Q^c$ and $D^c$ stand for randomized zero-error, quantum and deterministic communication complexity.
%\end{cor}
%\todo{Is this corollary interesting?  Equality already gives exponential separation b/w $Q$ and $D$.  Maybe only 
%state $Q_E$ version of this.}

%%%%%%%%%%%%%%%%%%%%%%%%%%%%%
%\section{One-sided randomized versus Las Vegas}
\section{Separations against Las Vegas complexity}
\label{sec:R1-R0}
%%%%%%%%%%%%%%%%%%%%%%%%%%%%%

In this section, we define a variant of the $f_{n,m}$ function from the last section.
Let $n$, $m$, $M$, $\tM$, $T$ and $T(j)$ be as previously.
The input alphabet is $\Sigma = \bool \times \tM \times \tM \times \tM$,
where we keep the names and notation of left, right and back pointers.  
Note that the back pointers now point to a cell of $M$, not a column.

Let $m$ be even.  The function $g_{n,m}\colon \Sigma^M\to\bool$ is defined like the function $f_{n,m}$ in \rf(sec:det) with condition 4 replaced by the following condition
\enumstart
	\item[$4'$]\label{4'} The set 
	$G = \sfigA{ j\in[m]\setminus\{b\} \midA \bpoint\sA[x_{\ell_j}] = a}$
	is of size exactly $m/2$.
\enumend

For an illustration refer to \rf(fig:1cert-tree-bpc-12_2).

\begin{figure}[H]
    \centering
    
    \begin{tikzpicture}[semithick, ->, scale=\scaling]
    
    \clip(-1,-1.25) rectangle (9,8); %otherwise the curved backpointers make canvas too large
    
    \tikzstyle{cell}=[rectangle,fill=black!10,draw=black,very thick,minimum size=\scaling cm,inner sep=0pt]
    
    \coordinate (lp) at  (-0.3, 0.0);
    \coordinate (rp) at  (+0.3, 0.0);
    \coordinate (bp) at  ( 0.0,-0.3);
    \coordinate (val) at ( 0.0, 0.10);
    
    \coordinate (botc) at (3.5,0);
    
    \draw[step=1cm,gray,very thin,-] (0, 0) grid (8,8);

    %1-column
    \foreach \y in {0, ..., 7} {
        \node[cell]    (c\y) at (3+0.5, \y+0.5) { };
        \node     ()  at ([shift={(val)}]c\y) {$1$};
    }
    \foreach \y in {0, ...,  5, 7} {
        \node    () at ([shift={(lp)}]c\y) {\scriptsize $\bot$};
        \node    () at ([shift={(rp)}]c\y) {\scriptsize $\bot$};
        \node    () at ([shift={(bp)}]c\y) {\scriptsize $\bot$};
    }
    \coordinate (spec) at (c6);
    \node    () at ([shift={(bp)}]spec) {\scriptsize $\bot$};

    %zeros
    \foreach \x/\y [count=\i] in 
    {0/1, 1/0, 2/1, 4/1, 5/1, 6/2, 7/3}
    {
        \node[cell]    (a\x) at (\x+0.5, \y+0.5) {$0$};
        %        \node          ()  at ([shift={(val)}]a\x) {$0$};
    }

    %tree inner nodes
    \foreach \x/\y [count=\i] in 
    {1/4, 5/5, 0/3, 2/3, 4/3, 6/4}
    {
        \node[cell]    (t\i) at (\x+0.5, \y+0.5) { };
    }

    %tree pointers
    \draw (spec) ++(lp) .. controls +(-180:1.25)  and +(90:1.0) .. (t1);
    \draw (spec) ++(rp) .. controls +( 15:0.5)  and +(90:1) .. (t2);
    \draw (t1) ++(lp) .. controls +(180:0.4)   and +( 90:0.75) .. (t3);
    \draw (t1) ++(rp) .. controls +(   0:0.4)  and +(90:0.75) .. (t4);
    \draw (t2) ++(lp) .. controls +(180:0.5)   and +( 90:1) .. (t5);
    \draw (t2) ++(rp) .. controls +( 0:0.4)    and +( 90:0.75) .. (t6);
    
    \draw (t3) ++(lp) .. controls +(-150:0.5)   and +( 120:1.25) .. (a0);
    \draw (t3) ++(rp) .. controls +(-15:0.5)   and +( 90:1) .. (a1);
    \draw (t4) ++(lp) .. controls +(-150:0.5)   and +(120:1.25) .. (a2);
    
    \draw (t5) ++(lp) .. controls +(-150:0.5)   and +(120:1.25) .. (a4);
    \draw (t5) ++(rp) .. controls +(-30:0.75)   and +( 90:1.0) .. (a5);
    \draw (t6) ++(lp) .. controls +(-150:0.5)   and +(120:1.25) .. (a6);
    \draw (t6) ++(rp) .. controls +( 0:0.4)   and +( 90:0.75) .. (a7);

    %backpointers
    \begin{scope}[very thin]
    \draw (a0) ++(bp) .. controls +(-30:0.75)  and +(-105:3.0) .. (c6);
    \draw (a2) ++(bp) .. controls +(-00:0.75)  and +(-95:3.0) .. (c6);
    \draw (a6) ++(bp) .. controls +(-150:0.75)  and +(-85:3.0)  .. (c6);
    \draw (a7) ++(bp) .. controls +(-150:0.75)  and +(-75:3.0)  .. (c6);

    %other backpointers
    \draw (a1) ++(bp) .. controls +(-15:0.75)  and +(-150:0.5) .. (2.98,0.5);
    \draw (a4) ++(bp) .. controls +(-90:0.5)  and +(-180:0.3)  .. (5, 0.5);
    \node    () at ([shift={(bp)}]a5) {\scriptsize $\bot$};
    \end{scope}

    \end{tikzpicture}
		\negbigskip
    
    \caption{An example of a $1$-certificate for the function $g_{8,8}$.  The tree $T$ is like in \rf(fig:tree) on the left.
        The center of a cell $x_{i,j}$ shows $\mathrm{val}(x_{i,j})$, the bottom of the cell shows $\mathrm{bpoint}(x_{i,j})$ and the bottom left and right sides show $\mathrm{lpoint}(x_{i,j})$ and $\mathrm{rpoint}(x_{i,j})$, respectively.
				Values and pointers that are not shown can be chosen arbitrarily.
				It is crucial that $m/2$ leaves point to the root $a$ of the tree, and $m/2-1$ leaves point to something different.
				}
    \label{fig:1cert-tree-bpc-12_2}
\end{figure}

\begin{thm}
\label{thm:R0-1.5}
If $n$ and $m$ are sufficiently large, the Las Vegas randomized query complexity $R_0(g_{n,m}) = \Omega(nm)$.
\end{thm}

\begin{proof}
We construct a hard probability distribution on negative inputs such that any Las Vegas randomized  algorithm has to make $\Omega(nm)$ queries in expectation to reject an input sampled from it.  
% Hence, for any Las Vegas algorithm, there exists a negative input, which requires an expected $\Omega(nm)$ number of queries to reject, proving $R_0(g_{n,m}) = \Omega(nm)$.
Each input $x = (x_{i,j})$ in the hard distribution is specified by a function $\ell_x\colon [m]\to[n]$.
The function specifies the positions of the leaves of the tree $T$ in a possible positive instance.
The definition of $x$ is as follows
\begin{equation}
\label{eqn:xijdef}
x_{i,j} =
\begin{cases}
(0,\bot,\bot,\bot), & \text{if $i = \ell_x(j)$;}\\
(1,\bot,\bot,\bot), & \text{otherwise.}
\end{cases}
\end{equation}
The hard distribution is formed in this way from the uniform distribution on all functions $\ell_x$.
Thus, all pointers are null pointers, and each column contains exactly one zero element in a random position.  The theorem obviously follows from the following two results.
\pfend

\begin{clm}
\label{clm:R0Needs}
Any Las Vegas algorithm for the function $g_{n,m}$ can reject an input $x$ from the hard distribution~\rf(eqn:xijdef) only if
it has found at least $m/2$ zeroes or
it has queried more than $n(m-1) - 2m$ elements.
\end{clm}

\begin{lem}
\label{lem:R0}
Assume a Las Vegas algorithm can reject an input $x$ from the hard distribution~\rf(eqn:xijdef) only if 
it has found $\Omega(m)$ zeroes or 
it has queried $\Omega(nm)$ elements.  
Then, the query complexity of the algorithm is $\Omega(nm)$.
\end{lem}

\pfstart[Proof of \rf(clm:R0Needs)]
Assume these conditions are not met.  Then, we can construct a positive input $y$ that is consistent with the answers to the queries obtained by the algorithm so far.

Let $B\subseteq[m]$ be the set of columns where no zero was found.  By assumption, $|B|\ge m/2+1$.  
Choose an element $b\in B$ and a subset $G\subseteq B\setminus\{b\}$ of size $m/2$.
Define $a = (\ell_x(b), b)$, set the value of $y_a$ to $1$ and its back pointer to $\bot$.
For each column $j\in G$, define $y_{\ell_x(j),j} = (0,\bot,\bot,a)$.
Finally, for the remaining columns $j\in B\setminus (G\cup\{b\})$, define $y_{\ell_x(j),j} = (0,\bot,\bot,\bot)$.

Remove from the tree $T$ the leaf with label $b$.  Let the resulting graph be $T'$.  
Put the root of $T'$ into $a$, and, for each $j\ne b$, put the leaf of $T'$ with label $j$ into $(\ell_x(j),j)$.  Put the remaining nodes of $T'$ into the still unqueried cells of $M$ preserving the structure of the graph.  Set their values to 0 and their back pointers to $\bot$.
Set all the remaining cells to $(1,\bot,\bot,\bot)$.  The resulting input $y$ is positive and consistent with the answers to the queries obtained by the algorithm.
\pfend

\pfstart[Proof of \rf(lem:R0)]
By \rf(thm:yao) it suffices to show that any deterministic algorithm $\cD$ makes an expected 
$\Omega(nm)$ number of queries to find $\Omega(m)$ zeroes in an input from the hard distribution.

Consider a node $S$ of the decision tree $\cD$.  Call a column $j\in[m]$ {\em compromised} in $S$ if either a zero was found in it, or more than $n/2$ of its elements were queried.
For an input $x$, let $A_t(x)$ be the number of compromised columns on input $x$ after $t$ queries.
Similarly, let $B_t(x)$ be the number of queries made outside the compromised columns.
Let us define
\[
I_t(x) = A_t(x) + \frac 2n B_t(x).
\]
Note that $A_t(x)$ can only increase as $t$ increases, whereas $B_t(x)$ can increase or decrease.

\begin{clm}
\label{clm:R0I}
For a non-negative integer $t$, we have 
\begin{equation}
\label{eqn:R0I}
\bE_x\skA[I_{t+1}(x)] - \bE_x\skA[I_t(x)] \le \frac 4n,
\end{equation}
 where the expectation is over the inputs in the hard distribution.
\end{clm}
 
\pfstart
Fix $t$.  
We say two inputs $x$ and $y$ are equivalent if they get to the same vertex of the decision tree after $t$ queries.  We prove that~\rf(eqn:R0I) holds with the expectation taken over each of the equivalency classes.
Fix an equivalence class, let $x$ be an input in the class, and $(i,j)$ be the variable queried by $\cD$ on the $(t+1)$st query on the input $x$.  Note that $(i,j)$, $A_t(x)$ and $B_t(x)$ do not depend on the choice of $x$.

Consider the following cases, where each case excludes the preceding ones.  All expectations and probabilities are over the uniform choice of an input in the equivalence class.
\itemstart
\item The $j$th column is compromised.  Then $I_{t+1}(x) = I_t(x)$, and we are done.
\item After the cell $(i,j)$ is queried, more than half of the cells in the $j$th column have been queried.
Then, $A_t(x)$ increases by $1$, and $B_t(x)$ drops by $\floor[n/2]$. 
Hence, $\bE_x\skA[I_{t+1}(x)] \le \bE_x\skA[I_t(x)]+1/n$.
\item Consider the remaining case.  We have $\mathrm{Pr}_x[i = \ell_x(j)]\le 2/n$.
If $i = \ell_x(j)$, then $A_t(x)$ grows by 1, and $B_t(x)$ can only decrease.
If $i\ne \ell_x(j)$, then $A_t(x)$ does not change, and $B_t(x)$ grows by 1.
Thus, $\bE_x\skA[I_{t+1}(x)] - \bE_x\skA[I_t(x)] \le \frac2n+\frac2n = \frac4n$.\qedhere
\itemend
\pfend

Now we finish the proof of \rf(lem:R0).
Assume the algorithm can reject an input $x$ only if it has found $c_1m$ zeroes or it has queried $c_2nm$ elements for some constants $c_1,c_2>0$.

Let $t = \floor[c_1 nm /8]$.  
Clearly, $\bE_x[I_0(x)]=0$ for all $x$.  
\rf(clm:R0I) implies $\bE_x[I_t(x)] \le c_1 m/2$.
By Markov's inequality, $\Pr_x\skA[I_t(x)\ge c_1 m]\le 1/2$.  
By our assumption, the probability the algorithm $\cD$ has not rejected $x$ after $t'=\min\{t,\; c_2 nm\}$ 
queries is at least $1/2$.  
Hence, the expected number of queries made by the algorithm is at least $t'/2 = \Omega(nm)$.
\end{proof}

\begin{thm}
\label{thm:R1}
For one-sided error randomized and exact quantum query complexity, we have $R_1(g_{n,m}) = \tO(n+m)$ and $Q_E(g_{n,m}) = \tO(n+m)$.
\end{thm}

\pfstart
For a positive input $x$, we call a column $j\in[m]$ \emph{good} iff it belongs to the set
$G$ from Condition~$4'$ on Page~\pageref{4'}.  Thus, a positive input has exactly $m/2$ good columns, whereas a negative one has none.
\rf(thm:R1) follows immediately from the following lemma.

\begin{lem}
\label{lem:good}
There exists a deterministic subroutine that, given an index $j\in[m]$, accepts iff the column $j$ is good in $\tO(n+m)$ queries.
\end{lem}

Indeed, given a string $y\in\bool^m$, it takes $O(1)$ queries for either an $R_1$ or a $Q_E$ algorithm to distinguish the case $y=0^m$ from the case when $y$ has exactly $m/2$ ones.  Using the subroutine from \rf(lem:good) as the input to this algorithm, we evaluate the function $g_{n,m}$ in $\tO(n+m)$ queries.  

For example, for an algorithm with one-sided error, we choose an index $j\in[m]$ uniformly at random, and execute the subroutine of \rf(lem:good).  If the input is negative, we always reject.  If the input is positive, we accept with probability exactly $1/2$.  The exact quantum algorithm is obtained similarly, using the Deutsch-Jozsa algorithm.
\pfend

\pfstart[Proof of \rf(lem:good)]
The subroutine is described in \rf(alg:good).
In the subroutine, $I$ stores the set of (the first indices) of the cells in column $j$ that can potentially contain the element $\ell_j$ back pointing to the special element $a$.
The set $B$ contains potentially marked columns.

\begin{algorithm}[tbhp]
\caption{A deterministic subroutine testing whether a column $j$ is good\label{alg:good}}
\medskip

\enumstart
\item Let $I\gets [n]$ and $B\gets [m]$.

\item While $I\ne\emptyset$ and $|B|\ge 2$, repeat the following:
\negmedskip
\enumstart
\item Let $i$ be the smallest element of $I$.  
Let $a\gets\bpoint(x_{i,j})$.
If $a=\bot$, remove $i$ from $I$, and continue with the next iteration of the loop.
\item Let $j$ be the smallest number of a column in $B$ that does not contain $a$.  Follow the pointers from $a$ as specified by the sequence $T(j)$.  Let $\ell_j$ be the endpoint.
\item If $\ell_j$ exists, is located in column $j$ and its value is 0, remove $j$ from $B$.  Otherwise, remove $i$ from $I$.
\enumend
\negmedskip

\item If $|B|\ge 2$, reject.  Otherwise, let $b$ be the only element of $B$.  Verify column $b$ using a procedure similar to that in \rf(alg:R0).
\item If column $b$ passes the verification, and $j$ belongs to the set $G$ from Condition~$4'$ on Page~\pageref{4'}, accept.  Otherwise, reject.
\enumend
\end{algorithm}

As ensured by step~4, the subroutine only accepts if the input is positive (column $b$ passes the verification), and the column $j$ is good.  Hence, we get no false positives.
On the other hand, assume the input is positive with the marked column $b$, column $j$ is good, and $\ell_j = (i,j)$.
In this case, $i$ cannot get removed from $I$ due to goodness of column $j$ and Condition~3 on Page~\pageref{3}, and $b$ never gets eliminated from $B$ as it contains no zeroes.
Thus, the only possibility to exit the loop on step~2 is to have $B=\{b\}$.  In this case, $b$ passes the verification, and the algorithm accepts since column $j$ is good.  Hence, we get no false negatives as well.

The query complexity of each iteration of the loop in step~2 is $O(\log m)$.  Also, with each iteration, either $I$ or $B$ get reduced by one element.  Hence, the total number of iterations of the loop does not exceed $n+m$.  Finally, the verification in step~3 requires $O(n+m)$ queries.  Thus, the query complexity of the algorithm is $\tO(n+m)$.
\end{proof}

\begin{cor}
\label{cor:R1-R0}
There is a total boolean function $f$ with $R_1(f) = \tO(R_0(f)^{1/2})$ and $Q_E(f) = \tO(R_0(f)^{1/2})$.
\end{cor}

\begin{proof}
We first obtain this separation for a non-boolean function.
Take $g_{n,m}$ with $n = m$.  The one-sided error randomized and exact quantum query complexity is $\tO(n)$ by \rf(thm:R1), and the Las Vegas query complexity is $\Omega(n^2)$ by \rf(thm:R0-1.5).
Since the size of the alphabet $\Sigma$ is polynomial, this also gives the separations for the associated boolean function $\tilde g_{n,n}$.
\end{proof}

%%%%%%%%%%%%%%%%%%%%%%%%%%
\section{Other separations against randomized complexity}
%%%%%%%%%%%%%%%%%%%%%%%%%%
\mycommand{ipoint}{\mathop{\mathrm{ipoint}}}
In this section we define another modification of the function used in \rf(sec:det).
Let $n$, $m$, $M$, $\tM$ and $T$ be as in \rf(sec:det), and let $k : 1\le k<m$ be an integer.
The new function $h_{k,n,m}\colon \Sigma^M\to\bool$ is defined as follows.
The input alphabet is $\Sigma = \bool \times \tM \times \tM \times \tM$.
For $v\in \Sigma$, we call the elements of the quadruple the \emph{value}, the \emph{left pointer}, the \emph{right pointer} and the \emph{internal pointer} of $x_{i,j}$, respectively.  We use notation $\val(v)$, $\lpoint(v)$, $\rpoint(v)$, and $\ipoint(v)$ for them in this order.

\renewcommand{\scaling}{0.9}
\begin{figure}[tbhp]
    \centering
    
    \begin{tikzpicture}[semithick, ->, scale=\scaling]
    
    %\clip(-1,-0.1) rectangle (9,10); %otherwise the curved pointers make canvas too large
    
    \tikzstyle{cell}=[rectangle,fill=black!10,draw=black,very thick,minimum size=\scaling cm,inner sep=0pt]
    
    \coordinate (lp) at  (-0.3, 0.0);
    \coordinate (rp) at  (+0.3, 0.0);
    \coordinate (ip) at  ( 0.0,+0.3);
    \coordinate (val) at ( 0.0,-0.10);
    
    \coordinate (botc) at (3.5,0);
    
    \draw[step=1cm,gray,very thin,-] (0, 0) grid (8,8);

    %1-columns
    \foreach \x in {2, 4, 5}
        \foreach \y in {0, ..., 7} {
            \node[cell]    (c\x\y) at (\x+0.5, \y+0.5) { };
            \node     ()  at ([shift={(val)}]c\x\y) {$1$};
        }
    
    \foreach \x in {2, 4}
    \foreach \y in {0, ...,  5, 7} {
        \node    () at ([shift={(lp)}]c\x\y) {\scriptsize $\bot$};
        \node    () at ([shift={(rp)}]c\x\y) {\scriptsize $\bot$};
        \node    () at ([shift={(ip)}]c\x\y) {\scriptsize $\bot$};
    }
    \foreach \y in {0, ...,  6} {
        \node    () at ([shift={(lp)}]c5\y) {\scriptsize $\bot$};
        \node    () at ([shift={(rp)}]c5\y) {\scriptsize $\bot$};
        \node    () at ([shift={(ip)}]c5\y) {\scriptsize $\bot$};
    }

    \coordinate (s1) at (c26);
    \coordinate (s2) at (c46);
    \coordinate (s3) at (c57);
%    \node    () at ([shift={(bp)}]spec) {\scriptsize $\bot$};

    %zeroes
    \foreach \x/\y [count=\i] in 
    {0/0, 1/1, 3/1, 6/2, 7/1}
    {
        \node[cell]    (a\x) at (\x+0.5, \y+0.5) {$0$};
        %        \node          ()  at ([shift={(val)}]a\x) {$0$};
    }

    %tree inner nodes
    \foreach \x/\y [count=\i] in 
    {1/4, 6/4, 0/2, 3/3, 7/3}
    {
        \node[cell]    (t\i) at (\x+0.5, \y+0.5) { };
    }

    %ipoint
    \begin{scope}[very thick]
    \draw   (s1) ++(ip) .. controls +( 30:0.5)  and +(150:0.75)  .. (c46);
    \draw   (s2) ++(ip) .. controls +( 0:0.5)  and +(-120:0.75)   .. (c57);
    \draw   (s3) ++(ip) .. controls +(90:1.0)  and +(105:1.75)   .. (c26);
    \end{scope}
    
%    %tree pointers

    \draw (s1) ++(lp) .. controls +( 180:0.5)  and +(105:2.0) .. (t1);
    \draw (s2) ++(lp) .. controls +(-135:2)  and +(75:2.0) .. (t1);
    \draw (s3) ++(lp) .. controls +( 165:2.5)  and +(90:1.75) .. (t1);
    \draw (s1) ++(rp) .. controls +(   0:1)  and +(105:2.0) .. (t2);
    \draw (s2) ++(rp) .. controls +( -15:1)  and +(90:1.0) .. (t2);
    \draw (s3) ++(rp) .. controls +( -15:0.5)  and +(75:0.75) .. (t2);
            
    \draw (t1) ++(lp) .. controls +(180:0.5)   and +( 90:1) .. (t3);
    \draw (t1) ++(rp) .. controls +( 15:1.5)  and +(90:1) .. (t4);
    \draw (t2) ++(rp) .. controls +(  0:0.4)   and +( 90:0.75) .. (t5);

    \draw (t3) ++(lp) .. controls +(180:0.5)   and +( 120:1.25) .. (a0);
    \draw (t3) ++(rp) .. controls +(  0:0.4)   and +( 90:0.75) .. (a1);
    \draw (t4) ++(rp) .. controls +(  0:0.5)   and +(60:1.25) .. (a3);
    \draw (t5) ++(lp) .. controls +(180:0.4)   and +( 90:0.75) .. (a6);
    \draw (t5) ++(rp) .. controls +( 0:0.5)   and +( 60:1.25) .. (a7);
    
    \end{tikzpicture}
    
    \caption{An example of a $1$-certificate for the $h_{3,8,8}$ function. 
		The tree $T$ is like in \rf(fig:tree) on the left.
        The center of a cell $x_{i,j}$ shows $\mathrm{val}(x_{i,j})$, the top of the cell shows $\mathrm{ipoint}(x_{i,j})$ and the left and right sides show $\mathrm{lpoint}(x_{i,j})$ and $\mathrm{rpoint}(x_{i,j})$, respectively.
				Values and pointers that are not shown can be chosen arbitrarily.
				}
    \label{fig:1cert-k-tree-bp}
\end{figure}
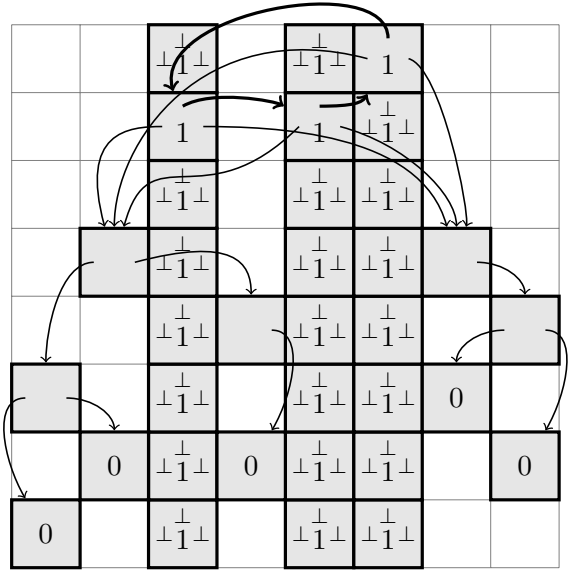

For an input $x = (x_{i,j})$, we have $h_{k,n,m}(x)=1$ if an only if the following conditions are satisfied (for an illustration refer to \rf(fig:1cert-k-tree-bp)):
\enumstart
\item There are exactly $k$ columns $b_1, \ldots, b_k$ such that $\val(x_{i,b_s})=1$ for all $i \in [n]$ and each $s \in [k]$.  We refer to these as the \emph{marked columns}.
	
\item Each marked column $b_s$ contains a unique cell $a_s$ such that $x_{a_s} \neq (1,\bot,\bot,\bot)$.  We call $a_s$ a \emph{special element}.

\item We have $\ipoint(x_{a_s}) = a_{s+1}$ for all $s\in[k-1]$, and $\ipoint(x_{a_k}) = a_1$.  Also, $\lpoint(x_{a_s}) = \lpoint(x_{a_t})$ and $\rpoint(x_{a_s}) = \rpoint(x_{a_t})$ for all $s,t\in[k]$.
	
\item For each non-marked column $j\in [m]\setminus\{b_1,\dots,b_k\}$, let $\ell_j$ be the end of the path which starts at a special element $a_s$ (whose choice is irrelevant) and follows the pointers $\lpoint$ and $\rpoint$ as specified by the sequence $T(j)$.
	We require that $\ell_j$ exists (no pointer on the path is $\bot$), $\ell_j$ is in the $j$th column, and $\val(x_{\ell_j}) = 0$.
\enumend

We use this function in two different modes.  
The first one is the $k=1$ case.  Then, $h_{1,n,m}$ is essentially the $f_{n,m}$ function from \rf(sec:det) with the back pointers removed, i.e., it need not satisfy condition~(4).  In this mode, the function is hard for a Monte Carlo algorithm, but still has low approximate polynomial degree.

The second mode is the general $k$ case.  In this mode, the function is hard for a Las Vegas algorithm, but feasible for quantum algorithms.  The proof of the lower bound is similar to \rf(sec:R1-R0).  We need a different function because \rf(alg:good) in \rf(thm:R1) cannot be efficiently quantized.

\begin{thm}
\label{thm:R0-2}
If $n$ and $m$ are sufficiently large, the Las Vegas randomized query complexity $R_0(h_{k,n,m}) = \Omega(nm)$ for any $k<m/2$.
\end{thm}

\pfstart
The proof is similar to the proof of \rf(thm:R0-1.5).
We define the hard distribution~\rf(eqn:xijdef) in exactly the same way.
The theorem follows from \rf(lem:R0) and the following claim.
\pfend

\begin{clm}
Any Las Vegas algorithm for the function $h_{k,n,m}$ can reject an input $x$ from the hard distribution~\rf(eqn:xijdef) only if
it has found $m-k+1$ zeroes or 
it has queried more than $n(m-k) - 2m$ elements.
\end{clm}

\pfstart
Assume these conditions are not met.  Then, we can construct a positive input $y$ that is consistent with the answers to the queries obtained by the algorithm so far.

Indeed, choose a set $B=\{b_1,\dots,b_k\}$ of columns where no zero was found.
Define $a_s = \sA[\ell_x(b_s), b_s]$ for each $s\in[k]$.  These elements have not been queried yet.
Define $\val(y_{a_s})=1$ for all $s$, as well as $\ipoint(y_{a_s}) = a_{s+1}$ for all $s\in[k-1]$, and $\ipoint(y_{a_k}) = a_1$.

Remove from the tree $T$ the leaves with labels in $B$ and the root.  Let the resulting graph be $T'$.  
For each $j\notin B$, put the leaf of $T'$ with label $j$ into $(\ell_x(j),j)$, set its value to 0 and all pointers to $\bot$.  Put the remaining nodes of $T'$ into the still unqueried cells of $M$ preserving the structure of the graph.
Set their value to 0 and their internal pointers to $\bot$. 
Let $u$ and $v$ be the cells where the left and the right child of the root of $T$ went.
For each $s\in [k]$, set $\lpoint(y_{a_s}) = u$ and $\rpoint(y_{a_s}) = v$.  Set all the remaining cells to $(1,\bot,\bot,\bot)$.  The resulting input is positive, and consistent with the answers to the queries obtained by the algorithm.
\pfend

\begin{thm}
\label{thm:R2}
If $n$ and $m$ are sufficiently large, the randomized query complexity $R(h_{1,n,m}) = \Omega\sA[\frac{nm}{\log m}]$.
\end{thm}

Note that \rf(thm:R2) only considers the case $k=1$.
It is proven in a similar fashion to \rf(thm:R0-2) using the additional fact that the expected size of a subtree rooted in a node of a balanced tree is logarithmic.  The proof is given in \rf(sec:R2proof).  

%%%%%%%%%%%%%%%%%%%%%%%
\mytxtcommand{defofh}{the definition of $h_{k,n,m}$}
\subsection{Quantum versus Las Vegas}

%\begin{thm}
%\label{thm:R1}
%The randomized query complexity with a one-sided error $R_1(h_{k,n,m}) = O(nm/k + kn + m)$.
%\end{thm}
%
%\begin{proof}
%We search for a column consisting only of ones.  In the positive case, there are $k$ such columns, so we will find one after $O(m/k)$ trials with high probability.  If we do not find such a column, we reject.
%Each trial takes $n$ queries, so we spend $O(nm/k)$ queries on this step.
%
%If we find a marked column $j$, we query all $x_{i,j}$ for $i \in [n]$ to check that it satisfies condition~(2) of \defofh.  Let $a$ be the corresponding special element.  We follow the internal pointer from $a$ to check that condition (3) is also satisfied and to find all the marked columns.  We then check condition (2) on them as well.  All this requires $kn$ queries.
%
%After that, we follow the left and right pointers from $a$, and check that condition (4) of \defofh is satisfied.  This requires $O(m)$ queries.
%
%We never accept a negative instance, so the algorithm has one-sided error.
%\end{proof}

\begin{thm}
\label{thm:Q-2}
The quantum query complexity $Q(h_{k,n,m}) = \tO\sA[\sqrt{nm/k}+\sqrt{kn}+k+\sqrt{m}]$.
\end{thm}

\pfstart
We search for a column consisting only of ones using Grover's search.  Testing one column takes $O(\sqrt{n})$ queries.  Also, in the positive case, there are $k$ such columns, so we will find one after $O(\sqrt{nm/k})$ queries with high probability.  If we do not find such a column, we reject.

In case we find a marked column $j$, we use Grover's search to check that it satisfies condition (2) of \defofh.  This requires $O(\sqrt{n})$ queries.  Let $a$ be the corresponding special element.
We follow the internal pointer from $a$ to find all the special elements and to check that condition (3) of \defofh is satisfied.  This requires $k$ queries.  We use Grover's search to check that all the remaining elements of the marked columns are equal to $(1,\bot,\bot,\bot)$.  This requires $O(\sqrt{kn})$ queries.

After that, we check condition (4) of \defofh.  
Since there are less than $m$ elements $\ell_j$ to check and each one can be tested in $O(\log m)$ queries, Grover's search can check this condition in $\tO(\sqrt{m})$ queries.
\pfend

\begin{cor}
\label{cor:Q-R0}
There is a total boolean function $f$ with $Q(f) = \tO(R_0(f)^{1/3})$.
\end{cor}

\begin{proof}
We first obtain the separation for a non-boolean function.
Take $h_{k,n,m}$ with $k=n$ and $m = n^2$.  Then the quantum complexity is $\tO(n)$ by \rf(thm:Q-2), and the Las Vegas randomized complexity is $\Omega(n^3)$ by \rf(thm:R0-2).
Since the size of the alphabet $\Sigma$ is polynomial, we obtain the required separations for the associated boolean function.
\end{proof}

%%%%%%%%%%%%%%%%%%%%%%%
\subsection{Exact Quantum versus Monte Carlo}

\begin{thm}
\label{thm:QE}
The exact quantum query complexity $Q_E(h_{k,n,m}) = O(n\sqrt{m/k} + kn + m)$.
\end{thm}

\begin{proof}
We use the exact version of Grover's search to find a column consisting only of ones.  In the positive case, there are exactly $k$ such columns, and testing each column takes $n$ queries.  Thus, the complexity of this step is $O(n\sqrt{m/k})$.

If we find a marked column $j$, we query all $x_{i,j}$ for $i \in [n]$ to check that it satisfies condition~(2) of \defofh.  Let $a$ be the corresponding special element.  We follow the internal pointer from $a$ to check that condition (3) is also satisfied and to find all the marked columns.  We then check condition (2) on them as well.  All this requires $kn$ queries.

After that, we follow the left and right pointers from $a$, and check that condition (4) of \defofh is satisfied.  This requires $O(m)$ queries.
\end{proof}

%\begin{cor}
%\label{cor:QE-R0}
%There is a total boolean function $f$ with $Q_E(f) = O(R_0(f)^{3/5})$.
%\end{cor}
%
%\begin{proof}
%Take $h_{k,n,m}$ with $k=\sqrt{n}$ and $m = n^{3/2}$.  Then the exact quantum query complexity is $O(n^{3/2})$ by \rf(thm:QE) and the Las Vegas randomized query complexity is $\Omega(n^{5/2})$ by \rf(thm:R0-2).
%Since the size of the alphabet $\Sigma$ is polynomial, this also gives the separation for the 
%associated boolean function $\tilde h_{k,n,m}$. 
%\end{proof}

\begin{cor}
\label{cor:QE-R2}
There exists a total boolean function $f$ with $Q_E(f) = \tO(R(f)^{2/3})$.
\end{cor}

\pfstart
Take $h_{1,n,m}$ with $m = n^2$.  Then the exact quantum query complexity is $O(n^2)$ by \rf(thm:QE) and the Monte Carlo randomized query complexity is $\widetilde\Omega(n^{3})$ by \rf(thm:R2).
Since the size of the alphabet $\Sigma$ is polynomial, this also gives the separation for the associated boolean function $\tilde h_{1,n,m}$.
\pfend

%%%%%%%%%%%%%%%%%%%%%%%
\subsection{Approximate Polynomial Degree versus Monte Carlo}
\begin{thm}
\label{thm:deg}
Let $\tilde h_{1,n,m}\colon \bool^{nm\ceil[\log|\Sigma|]}\to\bool$ be the boolean function associated to $h_{1,n,m}$.
The approximate polynomial degree $\tdeg(\tilde h_{1,n,m}) = \tO(\sqrt{n} + \sqrt{m})$
\end{thm}

\pfstart
For $j\in [m]$, let $g_j\colon \bool^{nm\ceil[\log|\Sigma|]}\to\bool$ be defined as follows.
The value $g_j(x)$ is 1 if $\tilde h_{1,n,m}(x)=1$, and $j$ is the marked column.  Otherwise, $g_j(x) = 0$.

For each $j\in[m]$, the function $g_j(x)$ can be evaluated in $\tO(\sqrt{n}+\sqrt{m})$ quantum queries using a variant of the VerifyColumn procedure in \rf(alg:Q).
Repeating this quantum algorithm $O(\log m)$ times, we may assume that its error probability is at most $1/(10m)$.
We then use the connection between quantum query algorithms and polynomial degree of \cite{beals:pol} to construct a
polynomial $p_j(x)$ of degree $2T$ (where $T$ is the number of queries) that is equal to the acceptance probability of this algorithm.
The polynomial $p_j(x)$ is of degree $\tO(\sqrt{n}+ \sqrt m)$ and satisfies 
$0\le p_j(x)\le 1/(10m)$ if $g_j(x)=0$, and $1-1/(10m)\le p_j(x)\le 1$ otherwise.

We then define a polynomial $p(x) = \sum_j p_j(x)$.
If $h_{1,n,m}(x)=0$, then all $g_j(x) = 0$ and $0\le p(x) \le 1/10$.
Otherwise, there is unique $b\in [m]$ such that $g_b(x)=1$, and all other $g_j(x)=0$.  
In this case, $1-1/(10m)\le p(x)\le 11/10$.  Thus, $p(x)$ is an approximating polynomial to $h_{1,n,m}$ and its degree is $\tO(\sqrt{n} + \sqrt m)$.
\pfend

\begin{cor}
\label{cor:deg-R2}
There is a total boolean function $f$ with $\tdeg(f) = \tO(R(f)^{1/4})$.
\end{cor}

\pfstart
Take $\tilde h_{1,n,m}$ from \rf(thm:deg) with $m = n$.  Then the approximate degree is $\tO(\sqrt{n})$, and the Monte Carlo randomized query complexity is $\Omega(n^2)$ by \rf(thm:R2).
\pfend

%%%%%%%%%%%%%%%%%%%%%%%%%%%%%%%%%
\section{Proof of \rf(thm:R2)}
%%%%%%%%%%%%%%%%%%%%%%%%%%%%%%%%%
\label{sec:R2proof}
\mycommand{cnst}{C_0}
\mycommand{treeexp}{{\cnst\log m}}
\mycommand{magic}{\frac{4\cnst \log m}n}

\begin{lem}
\label{lem:treeExp}
Assume $T$ is a balanced binary tree with $m$ leaves, and at least a $1/4$ fraction of its nodes are marked.
Let $u$ be sampled from all the marked nodes of $T$ uniformly at random.  The expected size of the subtree rooted at $u$ does not exceed $\treeexp$ for some constant $\cnst$.
\end{lem}

\pfstart
Let us first consider the case when $T$ is a complete balanced binary tree with $2^k-1$ nodes and all its nodes are marked.  Then, the expected size of the subtree is (where $i$ is the height of the node $u$)
\begin{equation}
\label{eqn:exp1}
\sum_{i=1}^{k} \frac{2^{k-i}}{2^k-1}\cdot (2^i-1) \le k.
\end{equation}

In the general case, $T$ can be embedded into a complete balanced binary tree $T'$ with $2^k-1$ nodes, where $k=\ceil[\log m]+1$.  Mark in $T'$ all the nodes marked in $T$.  Again, an $\Omega(1)$ fraction of the nodes is marked.  Hence, for each node $u$, a probability that $u$ is sampled from the marked nodes of $T'$ is at most a constant times its probability to be sampled from all the nodes of $T'$.
Thus, the expected size of the subtree is at most a constant times the value in~\rf(eqn:exp1).
\pfend

\mycommand{leaf}{\ell^{\mbox{\tiny L}}}
\mycommand{intern}{\ell^{\mbox{\tiny N}}}
\mycommand{allint}{T^{\mbox{\tiny N}}}

By \rf(thm:yao), it suffices to construct a hard distribution on inputs and show that any deterministic decision tree
 that computes $h_{1,n,m}$ with distributional error less than $2/10$ on the hard distribution 
makes an $\Omega(nm/\log m)$ expected number of queries.  
By Markov's inequality, it suffices to show that any deterministic decision tree that performs this task with error $3/8$ has \emph{depth} $\Omega(nm/\log m)$.
We now define the hard distribution.  

%By \rf(thm:yao) it suffices to construct a hard distribution on inputs and show that any deterministic algorithm 
%that computes $h_{1,n,m}$ on more than a $8/10$ fraction of inputs with respect to this distribution makes an $\Omega(nm/\log m)$ expected number of queries.  By Markov's inequality it suffices to show for a constant $c$
%that any decision tree of depth less than $cnm/\log m$  errs with probability at least $3/8$ on the hard distribution.
%We now define the hard distribution.  

Let $T$ be the balanced binary tree from the definition of $h_{1,n,m}$.
Denote by $r$ the root of the tree, and by $\allint$ the set of internal nodes of $T$.  The latter has cardinality $m-1$.

An input $x = (x_{i,j})$ is defined by a quadruple $(v_x, \pi_x, \leaf_x, \intern_x)$, where 
\itemstart
\item $v_x\in\bool$, it will be the value of the function $h_{1,n,m}$ on $x$;
\item $\pi_x\colon \allint\to [m]$ is an injection, it specifies to which columns the internal nodes of $T$ will go; and
\item $\leaf_x\colon [m]\to[n]$ and $\intern_x\colon [m]\to[n]$ are functions satisfying $\leaf_x(j)\ne \intern_x(j)$ for each $j\in[m]$.  They specify the rows where the leaves and the internal nodes of the tree $T$ land in column $j$.
\itemend
The definition is as follows.  Remove the leaf with the label $\pi_x(r)$ from $T$.  For each column $j\ne \pi_x(r)$, put the leaf $j$ into the cell $\sA[\leaf_x(j),j]$, set all its pointers to $\bot$ and its value to 0.
Then, for each internal node $u$ of the tree, put it at $\sA[\intern_x(\pi_x(u)), \pi_x(u)]$, set its left and right pointers so that the structure of the tree is preserved.
If $u\ne r$, assign its internal pointer to $\bot$, and set its value to 0.
Otherwise, if $u=r$, assign its internal pointer to itself, and set its value to $v_x$.
Assign the quadruple $(1,\bot,\bot,\bot)$ to all other cells.
It is easy to see that the value of the function $h_{1,n,m}$ on this input is $v_x$.

The hard distribution is defined as the uniform distribution over the quadruples $(v_x, \pi_x, \leaf_x, \intern_x)$ subject to the constraint $\leaf_x(j)\ne \intern_x(j)$ for each $j\in[m]$.
\medskip

\mycommand{depth}{D}
Let $\cD$ be a deterministic decision tree of depth 
\[
\depth = \frac{nm}{64\cnst \log m}.
\]
We will prove that $\cD$ errs on $x$, sampled from the hard distribution, with probability at least $3/8$.

Let $x$ be an input from the hard distribution, and consider the vertex $S$ of $\cD$ after $t$ queries to $x$.
We say that a column $j\in[m]$ and the corresponding tree element $\pi_x^{-1}(j)$ (if it exists) are \emph{compromised} on the input $x$ after $t$ queries if at least one of the following three conditions is satisfied:
\itemstart
\item one of the cells $(\leaf_x(j), j)$ and $(\intern_x(j), j)$ has been queried;
\item more than a half of the cells in the $j$th column have been queried; or
\item in the tree $T$ there exists an ancestor $u$ of $\pi^{-1}_x(j)$ such that one of the above two conditions is satisfied for $\pi_x(u)$.
\itemend
%Note that the algorithm can guess the value of the function with probability better than $1/2$ only if the root node of the whole tree has been queried.  The latter means that all the columns are compromised.

Let $A_t(x)$ denote the number of compromised columns, and $B_t(x)$ denote the number of cells queried outside the compromised columns, both after $t$ queries.
Consider the following quantity
\[
I_t(x) = \min\sfig{A_t(x) + \magic B_t(x),\;\; \frac m2}.
\]
Note that $A_t(x)$ can only increase as $t$ increases, whereas $B_t(x)$ can increase or decrease.

\begin{clm}
\label{clm:IExp}
For a non-negative integer $t$, we have 
\begin{equation}
\label{eqn:IExp}
\bE_x\skA[I_{t+1}(x)] - \bE_x\skA[I_t(x)] \le \frac{8\cnst \log m}n,
\end{equation}
 where the expectation is over the inputs in the hard distribution.
\end{clm}

We will prove \rf(clm:IExp) a bit later.  Now let us show how it implies the theorem.
Clearly, $I_0(x)=0$ for all $x$.
\rf(clm:IExp) implies that $\bE_x[I_D(x)] \le m/8$.
By Markov's inequality, the probability that $I_D(x)\ge m/2$ is at most $1/4$.

Let $x$ be an input satisfying $I_D(x)< m/2$.  Thus, $x$ has less than $m/2$ compromised columns after $D$ queries.  In particular, the variable $a = \sA[\intern_x(\pi_x(r)), \pi_x(r)]$, corresponding to the root of $T$, has not been queried.  
Let $y$ be the input given by $(1-v_x, \pi_x, \leaf_x, \intern_x)$.  
They only differ in $a$, and $h_{1,n,m}(x) \ne h_{1,n,m}(y)$.  Hence, the decision tree $\cD$ errs on exactly one of them.

This means that $\cD$ errs on $x$ sampled from the hard distribution with probability at least $3/8$.

\pfstart[Proof of \rf(clm:IExp)]
We divide the inputs of the hard distribution into equivalence classes, and prove that~\rf(eqn:IExp) holds with the expectation over each of the classes.
We say that two inputs $x$ and $y$ are equivalent if the following three conditions hold:
\itemstart
\item after $t$ queries, the decision tree $\cD$ gets to the same vertex on $x$ and $y$;
\item the set $C$ of compromised columns is the same in $x$ and $y$;
\item for all $j\in C$, $\pi^{-1}_x(j) = \pi^{-1}_y(j)$.
\itemend

Fix an equivalence class, let $x$ be an input in the class, and $(i,j)$ be the variable queried by $\cD$ on the $(t+1)$st query on the input $x$.  
Note that $(i,j)$, as well as $A_t(x)$ and $B_t(x)$ do not depend on the choice of $x$.
Consider the following cases, where each case excludes the preceding ones.
All expectations and probabilities are over the uniform choice of an input in the equivalence class.
\itemstart
\item We have $I_t(x) = m/2$.  Then, $I_{t+1}(x)\le m/2$, and we are done.
\item The $j$th column is compromised.  Then $I_{t+1}(x) = I_t(x)$, and we are done.
\item After the cell $(i,j)$ is queried, more than half of the cells in the $j$th column have been queried.
As $I_t(x)<m/2$, less than half of the columns are compromised.
By \rf(lem:treeExp) with non-compromised nodes marked, the expected growth of $A_t(x)$ is at most $\treeexp$.  On the other hand, the drop in $B_t(x)$ is at least $\floor[n/2]$.  Hence, $\bE_x\skA[I_{t+1}(x)] \le \bE_x\skA[I_t(x)]$.
\item Consider the remaining case.  We have $\mathrm{Pr}_x\skA[i \in \sfigA{\leaf_x(j), \intern_x(j)}]\le 4/n$.
\itemstart
\item If $i$ is one of $\leaf_x(j)$ or $\intern_x(j)$, then, as in the previous case, the expected growth of $A_t(x)$ is at most $\treeexp$, and $B_t(x)$ can only decrease.
\item If $i\ne a_x(j)$, then $A_t(x)$ does not change, and $B_t(x)$ grows by 1.
\itemend
Thus,
\[
\bE_x\skA[I_{t+1}(x)] - \bE_x\skA[I_t(x)] \le \frac4n\cdot \treeexp + \magic = \frac{8\cnst \log m}n.
\qedhere
\]
\itemend
\pfend

\section*{Acknowledgments}
This research is partially funded by the Singapore Ministry of Education and the
National Research Foundation, also through NRF RF Award No.\ NRF-NRFF2013-13, and 
the Tier 3 Grant ``Random numbers
from quantum processes,'' MOE2012-T3-1-009.
This research is also partially supported by the European Commission
IST STREP project Quantum Algorithms (QALGO) 600700, by the ERC Advanced Grant MQC, Latvian State Research Programme NeXIT project No. 1
and by the French ANR Blanc program under contract ANR-12-BS02-005 (RDAM project).

This work was done while A.B. was at the University of Latvia.  Part of it was also done while A.B. was at Centre for Quantum Technologies, Singapore.  A.B. thanks Miklos Santha for hospitality.

\bibliographystyle{\relativepath habbrvM}
%\bibliography{../../bib} %TexnicCenter doesn't get \relativepath :(
%\bibliography{\relativepath bib} %TexnicCenter doesn't get \relativepath :(
\bibliography{bib}

\end{document}